\documentclass[final,onefignum,onetabnum]{siamonline250211}

\ifpdf
  \hypersetup{
    pdftitle={Continuity of the Solution of a Non-Parametric Bayesian Statistical Calibration Procedure},
    pdfauthor={A. Prasadan, D. Estep, and D. Bingham}
  }
\fi

\headers{Continuity of Bayesian Statistical Calibration}{A. Prasadan, D. Estep, and D. Bingham}

\title{Continuity of a Non-Parametric Bayesian Statistical Calibration Procedure%
  \thanks{Submitted on \today.
  \funding{A. Prasadan acknowledges the support of the Natural Sciences and Engineering Research Council of Canada. D. Estep acknowledges the support of the Canada Research Chairs Program and  the Natural Sciences and Engineering Research Council of Canada. D. Bingham acknowledges the support of the Natural Sciences and Engineering Research Council of Canada.}}}

\author{
  Akshay Prasadan%
    \thanks{Department of Statistics and Actuarial Science, 
    Simon Fraser University, Burnaby, BC 
    (\email{akshay\_prasadan@sfu.ca}, \email{destep@sfu.ca}, \email{dbingham@sfu.ca}).}
  \and 
  Donald Estep\footnotemark[2]
  \and 
  Derek Bingham\footnotemark[2]
}

\usepackage{amssymb}
\usepackage{mathtools}
\usepackage{mathrsfs}
\usepackage{bbm}
\usepackage{enumitem}
\usepackage{xcolor}

\newtheorem{assumption}{Assumption}
\newcommand{\cA}{\mathcal{A}}
\newcommand{\cB}{\mathcal{B}}
\newcommand{\cC}{\mathcal{C}}
\newcommand{\cD}{\mathcal{D}}
\newcommand{\cK}{\mathcal{K}}
\newcommand{\cP}{\mathcal{P}}

\newcommand{\RR}{\mathbb{R}}
\newcommand{\NN}{\mathbb{N}}

\newcommand{\TV}{\mathrm{TV}}

\allowdisplaybreaks

\begin{document}
\maketitle

\begin{abstract}
    Recent work has developed a non-parametric Bayesian approach to the calibration of a computer model, which abstractly amounts to the inversion of a pushforward of stochastic input parameters by a smooth map. The framework has been used in several complex scientific applications, motivating our investigation on the continuity of the solution operator with respect to the distribution on the input parameters. We demonstrate that the solution operator for this approach is uniformly continuous in the total variation metric and weakly continuous for a broad class of distributions.
\end{abstract}

\begin{keywords}
Bayesian calibration, computer models, disintegration of measures, non-parametric Bayesian statistics,  stochastic inverse problems
\end{keywords}

% REQUIRED
\begin{MSCcodes}
60A10, 62G05,  62P30
\end{MSCcodes}

\tableofcontents

\section{Introduction}
\label{section:introduction}

Computer models, or simulators, have become indispensable tools for exploring and designing physical systems \cite[e.g.]{grosskopf2020, higdon_2008, kennedy_ohagan}. Frequently,  the simulator response is governed by physical parameters. For instance, in storm surge modeling, parameters could correspond to physical properties of the ocean floor, and specifying their values permits a simulator to generate predictions of flooding behavior. However, these parameters are not always known and must be inferred using field observations. This gives rise to a calibration problem of inferring parameters of the simulator using field observations. There are many methodologies for this problem, however, an important feature of any methodology is stability of the solution with respect to the input, and this is the focus of our paper. 

An important class of problems treats the situation where the parameters of the simulator vary as trials are conducted, contributing to the stochasticity in the field observations. In that case, an inferential object is a distribution on the input parameters. A series of papers have adopted a non-parametric Bayesian approach to this calibration problem, based on the theory of disintegration of measure \cite{SIP_I_2011, SIP_II_2012,sip_III_2014, esip_2025}. The approach has been referred to as the stochastic inverse problem, data-consistent inversion, and the Statistical Calibration Problem (SCP); we adopt the latter name. The SCP updates a prior distribution with field obervations to produce a posterior distribution on the parameters. Applications of this method have included storm surges, vibration-induced material damage,  groundwater contamination, and Covid 19 outbreaks \cite{manning_2015, butler_vibrational_2015, covid_2025, groundwater_2015}. 

A simple example of an SCP is illustrated in \cref{fig:intro_example}, where the plot on the right is the observed output distribution of field observations, while the plot on the left shows the unobserved distribution of the input parameters. We call the distribution of the input parameters the `trial generating distribution' (TGD) since each field observation is generated by a single draw of the parameter, i.e., a single trial. We cannot (in general) directly estimate the TGD on the left plot, because there are many distributions that generate the same field observations. This motivates us to rely on prior information to make a choice of distribution that remains consistent with the field data.  

A more involved example is given in \cite{ip_review_2024}, where the authors examine the temperature profile of a heated material. The associated SCP is to infer the thermal diffusivity parameters that feed into the heat equation and induce the observed temperature at a particular elapsed time; stochasticity in the input arises from using different samples of the material. Importantly, the model here is many-to-one, since multiple combinations of diffusivity parameters map to the same temperature. An observed temperature value then corresponds to a contour of diffusivity parameters in the input space, along which the physical model yields the same output. 

 The  SCP can be viewed as an empirical stochastic inverse problem for a random vector of field data. A solution to the SCP is a distribution on the input space whose induced output distribution (or pushforward measure) matches the distribution of observed field data. That is, stochastic input parameters are passed through a map to generate a random vector, and we aim to invert this process. To solve the problem,  \cite{SIP_I_2011, SIP_II_2012,sip_III_2014, esip_2025} use a numerical approximation of a disintegration of measure. In short, one uses the field data to implicitly construct probability distributions on the contours of a many-to-one map, informed by a prior, and this family of distributions collectively induces an estimate of the input distribution. The proposed estimator is proven to be a valid solution to the SCP because its pushforward  distribution by the map matches the observed distribution.

As mentioned earlier, a desireable property of such a solution is continuity in the underlying input distribution that generated the field data. That is, if the input distribution is perturbed slightly, then the solution should be perturbed only slightly. This of both theoretical and practical importance. First, the input distribution may evolve over time or space in real world experiments, e.g., if samples are taken over a wide time interval or geography. For example, the vegetation of coastal settings may change over time in the coastal surge example. Second, an experiment may be subject to random effects that change in distribution across different iterations. For instance, a lab collecting biological samples may improve its method of sterilizing tools, changing the distribution of inputs into machines that perform various measurements. In particular, the random influence of unwanted specimen may disappear over time, and the SCP solution should adapt continuously to this change.

 Another implication of continuity is convergence of the SCP solution for sequences of TGDs. For example,  in the ground contamination application of \cite{groundwater_2015}, one may be interested in unknown geological parameters that are fixed on any reasonable time scale, as opposed to ones that change with each sample. Treating the unknown input parameter as random with a very tight variance should yield similar results as solving an inverse problem of a deterministic parameter. Thus, the SCP procedure unifies Bayesian calibration in two regimes: when the parameters of interest are random across trials and when the parameters are fixed. Furthermore, continuity is useful for analyzing large numbers of trials of an experiment. If $n$ unobserved samples are drawn from the input distribution to generate the output sample, then consider a smoothed density estimator fitted to the input data. As $n$ grows large, the smoothed density estimator converges to the  input distribution. The SCP estimator should show close agreement between application to a smoothed empirical density as input versus the original input distribution.

In this paper, we prove this continuity property for the estimator proposed by \cite{esip_2025}. We show that their SCP solution is stable to  perturbations in the input distribution, and when a sequence of input distribution converges, the SCP solution also converges. These continuity results differ from those already in \cite{esip_2025}, in which the authors demonstrate that for a fixed input distribution, the solution has a continuous density function over the input space. Stability was also demonstrated in \cite{Butler_2018} for an earlier version of the SCP, in which it is shown that the solution is continuous in the total variation metric. However, our approach will use surface integral representations developed in \cite{esip_2025}, which avoids a direct application of a Bayes' Theorem on measure 0 contours. Moreover, we are able to prove weak convergence results that permit a much broader class of input distributions, violating the assumptions required in \cite{Butler_2018}.

More technically, our contributions are as follows: First, we show the SCP solution operator is a non-expansive, uniformly continuous function of the TGD in the total variation (TV) metric. We use this result to prove pointwise convergence almost everywhere (a.e.) of the TGD density implies TV convergence of the SCP solution.  We then prove that the SCP solution operator is weakly continuous, i.e., we show that weak convergence of the trial-generating probability measure implies weak convergence of the SCP probability measure. Both results are over the space of probability measures on the input space whose pushforward measure has a density with respect to Lebesgue measure. Notably, the weak convergence result extends to the setting where the limiting input distribution is a Dirac measure or a finite mixture of Dirac measures. We then illustrate this convergence behavior with a simulation. These properties are further evidence of the usefulness of the SCP algorithm and its extensions to scientific problems.

\section{Review of the SCP} 
\label{section:background}

We briefly summarize some of the notation used throughout the paper.  For any set $A$, let $\cB_A$ denote the Borel $\sigma$-algebra restricted to $A$. If $\nu_1,\nu_2$ are two measures, we write $\nu_1\ll \nu_2$ if $\nu_1$ is absolutely continuous with respect to $\nu_2$. If $\phi$ is a measurable function and $\nu$ a measure on a measurable space $(\cA,\cB_{\cA})$, we write $\phi\nu$ for the pushforward measure of $\nu$ by $\phi$. The pre-image of $\phi$ over a set $B$ in the co-domain is denoted $\phi^{-1}(B)$. We write  $\mu_A$ to be the Lebesgue measure over a measurable subset $A$ of Euclidean space. The indicator function for a set $A$ is denoted $\chi_A$. 

We introduce definitions and notation similar to those in \cite{esip_2025}. The quantities we introduce are also listed in \cref{table:1} for convenience.   Let $(\Lambda,\cB_{\Lambda})$ be a measurable space with $\Lambda\subset \RR^n$, corresponding to the physical parameter space. Let $Q\colon (\Lambda,\cB_{\Lambda})\to (\cD,\cB_{\cD})$ be a measurable map with $Q(\Lambda)=\cD\subset\RR^m$, representing for example a computer model. Let $P_{\Lambda}^t$ be the unknown trial-generating measure on $\Lambda$ which produces field observations distributed as $P_{\cD}=QP_{\Lambda}^t$. The goal of the SCP: Given $P_{\cD}$, find a probability measure $\hat P_{\Lambda}$ such that $P_{\cD}=Q\hat P_{\Lambda}$. An empirical formulation of the SCP is also of interest: given a sample $\{q_i\}$ drawn from $P_{\cD}$, find $\hat P_{\Lambda}$ such that $\{q_i\}$ is a sample from $Q\hat P_{\Lambda}$. In the example of \cref{fig:intro_example}, the left-plot is the unknown trial-generating distribution $P_{\Lambda}^t$ that induces the observed $P_{\cD}$ on the right-plot. We henceforth assume $Q$ is not injective, corresponding to situations where the experiment provides limited information about the physical system.

\begin{figure}[t]
    \centering
    \includegraphics[width=0.48\textwidth]{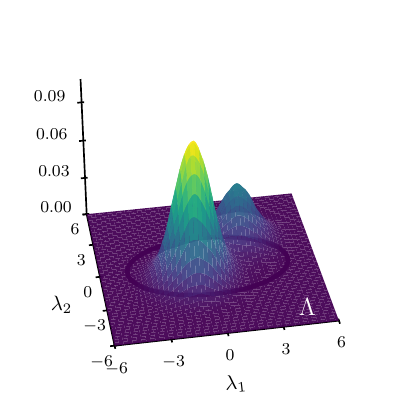}
    \includegraphics[width=0.48\textwidth]{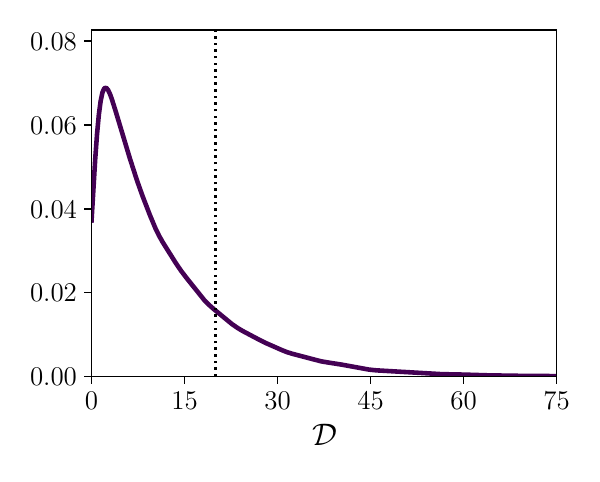}
    \caption{We illustrate an example of  distribution $P_{\Lambda}^t$ on a set $\Lambda\subseteq\RR^2$, and the pushforward distribution $QP_{\Lambda}^t$ by a map $Q\colon \Lambda\to\cD$. We set $Q(\lambda_1,\lambda_2)=\lambda_1^2+3\lambda_2^2$ and $P_{\Lambda}^t$ to be an unequally weighted mixture of two bivariate Gaussians. Left: Kernel density estimate (KDE) fitted to data drawn from $P_{\Lambda}^t$. The solid black line indicates the contour $Q^{-1}(20)$. Right: KDE of the pushforward distribution on $\cD$ generated by applying $Q$ to points $\lambda$ drawn from $P_{\Lambda}^t$. The dotted line corresponds to $Q=20$.}
    \label{fig:intro_example}
\end{figure}

The mathematical tool  for investigating solutions of the SCP in \cite{esip_2025} is  disintegration of measures \cite{ramble,chang_pollard}. Briefly stated, given any solution $\hat P_{\Lambda}$ to the SCP (i.e., such that $Q\hat P_{\Lambda}=P_{\cD}$), there is a family of probability measures $\{P(\,\cdot\,| q)\}_{q\in\cD}$ on $(\Lambda,\cB_{\Lambda})$ that are unique $P_{\cD}$ a.e. that concentrate on the pre-image $Q^{-1}(q)$, in the sense that $P(Q^{-1}(q)| q)=1$ while $P(\Lambda\setminus Q^{-1}(q)| q)=0$. The $P(\,\cdot\,| q)$ are  regular conditional probability measures.   Moreover, we have the following property, referred to as a disintegration of $\hat P_{\Lambda}$ by $Q$: \begin{equation} \label{eq:distintegration}
    \hat P_{\Lambda}(A)=\int_{Q(A)}\int_{Q^{-1}(q)\cap A} dP(\lambda| q) dP_{\cD}(q), \quad A\in\cB_{\Lambda}.
\end{equation} Conversely, given a family of measures $\{P(\,\cdot\,| q)\}_{q\in\cD}$ that concentrate on the contour $Q^{-1}(q)$, one obtains a unique $\hat P_{\Lambda}$ with the aforementioned disintegration property.  For instance, in the left plot of \cref{fig:intro_example}, we depict the contour $Q^{-1}(20)$ as the elliptical black curve; the associated regular conditional probability measure $P(\,\cdot\,| 20)$ is a probability distribution on that ellipse, which notably has $\mu_{\Lambda}$-measure 0. 

The property \cref{eq:distintegration} suggests a means of deriving a solution $\hat P_{\Lambda}$ by starting with information on the right-hand side, especially since $P_{\cD}$ is observable or can be estimated.  However, the family of measures $\{P(\,\cdot\,| q)\}_{q\in\cD}$ that corresponds to the  trial-generating distribution $P_{\Lambda}^t$ is  inherently unknowable because $P_{\Lambda}^t$ is unknown, so we make a subjective choice via an Ansatz. An Ansatz is an assumption of a family of regular conditional probabilities on the contours of $Q$, which can be chosen, for example, by disintegrating any suitable measure on $\Lambda$ by $Q$. In \cite{esip_2025}, the authors begin with a prior $P_p$ measure on $(\Lambda,\cB_{\Lambda})$ and disintegrate to obtain  an Ansatz family $\{P_{p}(\,\cdot\,| q)\}_{q\in\cD}$. For instance, the authors suggest disintegrating the uniform prior by appealing to information theoretic arguments. 

With any prior and its implied Ansatz, one obtains a unique solution to the SCP (formally speaking, the regular conditional measure $P_{p}(\,\cdot\,| q)$ is substituted into \cref{eq:distintegration}, but in practice this is unnecessary). In particular, under quite general assumptions, given a prior $P_p$ on $(\Lambda,\cB_{\Lambda})$ and pushforward $P_{\cD}$ on $(\cD,\cB_{\cD})$, a unique solution $\hat P_{\Lambda}$ to the SCP exists, and if the prior and pushforward both have densities with respect to the respective Lebesgue measure, then so does $\hat P_{\Lambda}$. In general, however,  $\hat P_{\Lambda}$ does not match the unknown trial-generating distribution $P_{\Lambda}^t$, since we do not have access to the disintegration of $P_{\Lambda}^{t}$ by $Q$.  

To estimate the SCP solution, an importance sampling algorithm can be used, as in \cite{esip_2025}, based on the approximation \[\hat P_{\Lambda}(A)\approx \sum_i P_p(A | Q^{-1}(D_i))\cdot P_{\cD}(D_i),\] where the sets $D_i$ form a partition of $\cD$. The $P_{\cD}(D_i)$ factor can be estimated using a histogram estimate of the observed  samples. For $P_p(A | Q^{-1}(D_i))$, one generates $\lambda_1,\dots,\lambda_J\sim P_p$, isolates the subset with $Q(\lambda_j)\in D_i$, and computes the proportion of that subset in $A$. An explicit formula is given in \cite[Equations (4.3), (4.4), and (4.5)]{esip_2025}. We discuss the example in more detail later, but we illustrate three instances of SCP solutions in \cref{fig:final_example}, where in the top row and bottom left plots we display the solution as a heatmap. We overlay the TGD in orange over these three plots. The bottom-right plot shows an example prior $P_p$.

The numerical convergence relies on several regularity conditions, which are also essential to work with concrete representations of the disintegration in terms of densities. Some of these assumptions are a consequence of relying on a non-parametric approach. Thus, we restate the assumptions given in \cite{esip_2025}.

\begin{assumption} \label{assumption:compact}
$\Lambda$ is compact, has measurable boundary $\partial\Lambda$ of null measure, and the closure of $\mathrm{int}(\Lambda)$ is $\Lambda$.
\end{assumption}

\begin{assumption} \label{assumption:Q:smooth}
$Q$ is continuously differentiable a.e. on an open set $U_{\Lambda}$ containing $\Lambda$.
\end{assumption}

\begin{assumption} \label{assumption:Q:rank}
The Jacobian $J_Q$ of $Q$ is of full rank $m$ except on finitely many manifolds of dimension $\le n-1$.
\end{assumption}

\begin{assumption} \label{assumption:prior}
The prior measure $P_p$ has a density $\rho_p$ with respect to $\mu_{\Lambda}$, and for $\mu_{\cD}$-almost all $q\in\cD$, if $\mu_N(\,\cdot\,|  q)$ is the regular conditional measure obtained by disintegrating $\mu_{\Lambda}$, then $\mu_N(\rho_p=0| q)\ne 1$.
\end{assumption}

\begin{assumption} \label{assumption:pushforward}
The TGD $P_{\Lambda}^{t}$ has a pushforward measure $P_{\cD}=QP_{\Lambda}^t$ with density $\rho_{\cD}$ with respect to $\mu_{\cD}$.
\end{assumption}

We require Assumptions \ref{assumption:compact}, \ref{assumption:Q:smooth}, \ref{assumption:Q:rank}, and \ref{assumption:prior} in every result of this paper, as these are core to the SCP algorithm. Assumption \ref{assumption:pushforward} directly pertains to the TGD, unlike Assumptions \ref{assumption:compact}-\ref{assumption:prior}, and it is required to apply the SCP operator to an input distribution. However, later in the paper, we consider sequences of TGDs and their limiting distributions. As a consequence, there is some flexibility in imposing Assumption \ref{assumption:pushforward}, as our techniques may not require running the SCP algorithm on the limiting distribution itself. No such flexibility is permitted with Assumptions \ref{assumption:compact}-\ref{assumption:prior} since we hold fixed $\Lambda,Q, P_p$ in any sequence of TGDs.

To explain Assumption \ref{assumption:pushforward}, we introduce the following notation. Let $\mathscr{P}_{\Lambda,Q}$ be the set of probability measures on $\Lambda$ whose pushforward by $Q$ is absolutely continuous with respect to $\mu_{\cD}$. Then Assumption \ref{assumption:pushforward} is equivalent to requiring $P_{\Lambda}^t\in\mathscr{P}_{\Lambda,Q}$. We are always explicit about using Assumption \ref{assumption:pushforward}  using this notation, while we may omit repetition of Assumptions \ref{assumption:compact}-\ref{assumption:prior}.

A simple sufficient condition for Assumption \ref{assumption:pushforward} is that $P_{\Lambda}^t\ll \mu_{\Lambda}$, i.e., that there is a density for the TGD. To see this, note that this condition implies $QP_{\Lambda}^t\ll Q\mu_{\Lambda}$, and it is known from \cite[Theorem E.2]{esip_2025} that $Q\mu_{\Lambda}\ll \mu_{\cD}$, with the density being a surface integral $\tilde\rho_{\cD}$, defined in \cref{table:1}. Hence $QP_{\Lambda}^t \ll \mu_{\cD}$, i.e., $P_{\Lambda}^t\in\mathscr{P}_{\Lambda,Q}$ as desired.

\begin{table}
\centering
\renewcommand{\arraystretch}{1.3}
\begin{tabular}{ll}
\textbf{Symbol} & \textbf{Meaning or Use} \\
\hline
$\Lambda$, $\cD$ & Input space in $\RR^n$ and output space $Q(\Lambda)$ in $\RR^m$ \\
$\mu_{\Lambda}$, $\mu_{\cD}$ & Lebesgue measures on $\Lambda$ and $\cD$, respectively \\
$\mu_{N}(\,\cdot\, | q)$ & Regular conditional measure of $\mu_{\Lambda}$ given $Q=q$ \\
$P_{\Lambda}^{t}$ or $P_{\Lambda}^{t,i}$; $\hat p_{\Lambda}^t$ or $\hat p_{\Lambda}^{t,i}$ & TGD on $\Lambda$ and density in $\mu_{\Lambda}$ \\
$\hat{P}_{\Lambda}$ or  $\hat{P}_{\Lambda}^i$; $\hat p_{\Lambda}$ or $\hat p_{\Lambda}^i$ & SCP solution and density in $\mu_{\Lambda}$ \\
$\mathscr{P}_{\Lambda,Q}$ & Set of TGDs $P_{\Lambda}^{t}$ with a pushforward $QP_{\Lambda}^t\ll \mu_{\cD}$\\
$\overline{\mathscr{P}_{\Lambda,Q}}$ & Sequential weak closure of $\mathscr{P}_{\Lambda,Q}$  \\
$P_p$; $\rho_p$ & Prior on $\Lambda$ and its density with respect to $\mu_{\Lambda}$ \\
$P_p(\,\cdot\, | q)$ & Regular conditional probability for $P_p$ given $Q=q$ \\
$\rho_{\cD}(q)$ or $\rho_{\cD,i}(q)$ & Pushforward density of TGD: $\frac{d(Q P_{\Lambda}^t)}{d\mu_{\cD}}$ or $\frac{d(Q P_{\Lambda}^{t,i})}{d\mu_{\cD}}$ \\
$\tilde\rho_{\cD}(q)$ & Scaling Factor: $\frac{d(Q\mu_{\Lambda})}{d\mu_{\cD}} = \int_{Q^{-1}(q)} \frac{1}{\sqrt{\det J_Q J_Q^T}} \, ds$ \\
$\tilde \rho_{p,\cD}(q)$ & Scaling factor: $\frac{d(Q P_p)}{d\mu_{\cD}} = \int_{Q^{-1}(q)} \frac{\rho_p}{\sqrt{\det J_Q J_Q^T}} \, ds$ \\
\end{tabular}
\caption{Summary of common notation. TGD refers to the trial-generating distribution on $\Lambda$. Indexing by $i$ indicates a sequence of distributions or densities.}
\label{table:1}
\end{table}

\section{Total Variation Continuity} \label{section:total:variation}

We prove a continuity theorem in the total variation metric on $\Lambda$. A similar stability  result was already developed in \cite[Theorem 4.1]{Butler_2018}, as previously mentioned. Our approach uses surface integral representations of regular conditional probability measures from \cite{esip_2025}. In this section, we require Assumptions  \ref{assumption:compact}, \ref{assumption:Q:smooth}, \ref{assumption:Q:rank},  \ref{assumption:prior}, and any TGD, limiting or not,  must satisfy Assumption \ref{assumption:pushforward}. 

We first establish some important definitions. If $P_1$ and $P_2$ are two measures on $\Lambda$, the total variation distance $d_{\TV}(P_1,P_2)$ between $P_1$ and $P_2$ is defined as $\sup_{E\in\cB_{\Lambda}}|P_1(E)-P_2(E)|$. If $P_1$ and $P_2$ have densities $\rho_1,\rho_2$ with respect to a measure $\mu$, then \[d_{\TV}(P_1,P_2)=d_{\TV}(\rho_1,\rho_2):= \frac{1}{2}\int |\rho_1(x)-\rho_2(x)|d\mu(x).\] Now let  $P_{\cD}$ be the pushforward by $Q$ of the unknown TGD $P_{\Lambda}^t\in\mathscr{P}_{\Lambda,Q}$, and let $P_p$ be a prior distribution on $\Lambda$ with density $\rho_p$ with respect to $\mu_{\Lambda}$. Let  $\rho_{\cD} := \frac{d P_{\cD}}{d\mu_{\cD}}= \frac{d QP_{\Lambda}^t}{d\mu_{\cD}}$, which exists by definition of $\mathscr{P}_{\Lambda,Q}$, and $    \tilde\rho_{p,\cD}(q) := \frac{d QP_p}{d\mu_{\cD}}(q).$ It is shown in \cite{esip_2025} that  \[\tilde\rho_{p,\cD}(q)=\int_{Q^{-1}(q)} \frac{\rho_p}{\sqrt{\det(J_Q J_Q^{T})}} ds,\]  which is a surface integral over $Q^{-1}(q)$ using some piece-wise smooth parametrization. By \cite[Theorem 3.5]{esip_2025}, there is a solution to the SCP $\hat P_{\Lambda}$ with density $\hat \rho_{\Lambda}$ with respect to $\mu_{\Lambda}$ such that for $\mu_{\Lambda}$ almost all $\lambda$, \begin{align} \label{eSIP:density:solution}
    \hat\rho_{\Lambda}(\lambda) =  \frac{\rho_{\cD}(Q(\lambda)) \rho_p(\lambda)}{\tilde\rho_{p,\cD}(Q(\lambda))}.
\end{align} 

Our first result, \cref{theorem:stability}, proven in the appendix, establishes that the SCP solution is uniformly continuous and non-expansive on the trial-generating space $\mathscr{P}_{\Lambda,Q}$ with respect to the TV metric. In practical terms, the theorem ensures that the SCP solution changes in a controlled manner when the input distribution is perturbed, say due to slight shifts in experimental conditions. An analogous result appears in \cite[Theorem 4.1]{Butler_2018}.

\begin{theorem}[TV Stability] \label{theorem:stability} Let $P_{\Lambda}^t$ and $R_{\Lambda}^t$ be two measures in $\mathscr{P}_{\Lambda,Q}$ with pushforward densities $p_{\cD}=\frac{dQP_{\Lambda}^t}{d\mu_{\cD}}$ and $r_{\cD}=\frac{dQR_{\Lambda}^t}{d\mu_{\cD}}$. Let $P_p$ be a prior distribution on $\Lambda$ with density $\rho_p = \frac{dP_p}{d\mu_{\Lambda}}$, and let $\hat p_{\Lambda}$ and $\hat r_{\Lambda}$ be the densities with respect to  $\mu_{\Lambda}$ of the SCP solutions for $P_{\Lambda}^t$ and $R_{\Lambda}^t$, respectively. Then \[d_{\TV}(\hat p_{\Lambda},\hat r_{\Lambda})  = d_{\TV}\left( p_{\cD}, r_{\cD}\right) \le  d_{\TV}\left( P_{\Lambda}^t,R_{\Lambda}^t\right).\]
\end{theorem}

The next result is one of several in which we consider a sequence of unknown TGDs $\{P_{\Lambda}^{t,i}\}_{i\in\NN}$ that converge in a specified metric. Our claim is that the  sequence of corresponding SCP solution measures $\hat P_{\Lambda}^i$ also converges in a specified metric. The following local limit theorem proves that pointwise convergence almost everywhere of a sequence of trial-generating densities implies TV-convergence of the SCP solutions. Note that this theorem requires Assumption \ref{assumption:pushforward}  for all $i\in\NN$ and the limiting TGD.

\begin{theorem}[Local Limit Theorem] \label{theorem:local:limit}
    Let $\{P_{\Lambda}^{t,i}\}_{i\in\NN}$ be a sequence of TGDs in $\mathscr{P}_{\Lambda,Q}$ with densities $p_{\Lambda}^{t,i}$ with respect to $\mu_{\Lambda}$. Let $P_{\Lambda}^{t,\ast}$ be a TGD in $\mathscr{P}_{\Lambda,Q}$ with density $p_{\Lambda}^{t,\ast}$ with respect to $\mu_{\Lambda}$. Let $P_p$ be any prior on $\Lambda$ with density $\rho_p$ with respect to $\mu_{\Lambda}$. Let $\{\hat P_{\Lambda}^i\}$ be the corresponding sequence of SCP solutions to the TGD sequence, with densities $\hat p_{\Lambda}^i$ with respect to $\mu_{\Lambda}$. Let $\hat P_{\Lambda}^{\ast}$ be the solution of the SCP for $P_{\Lambda}^{t,\ast}$, with density $\hat p_{\Lambda}^{\ast}$ in $\mu_{\Lambda}$. Suppose $p_{\Lambda}^{t,i}(\lambda)\to p_{\Lambda}^{t,\ast}(\lambda)$ for $\mu_{\Lambda}$-a.e. $\lambda$.  Then $d_{\TV}(\hat p_{\Lambda}^{i},\hat p_{\Lambda}^{\ast})\to 0$.
\end{theorem}
    \begin{proof}
            By Scheff\'e's theorem, $d_{\TV}(p_{\Lambda}^{t,i},p_{\Lambda}^{t,\ast})\to 0$. By  \cref{theorem:stability}, $d_{\TV}(\hat p_{\Lambda}^{i},\hat p_{\Lambda}^{\ast})\to 0$.
    \end{proof}

This theorem opens up several  applications we discussed in the introduction. For example, we may have a sequence of TGDs influenced by a random effect with decaying variance across different iterations of an experiment, say due to increasing precision of experimental controls.

\section{Weak Continuity}
\label{section:weak:continuity}

Next, we derive conditions under which weak convergence of the TGD gives weak convergence of the SCP solution. Our results effectively establish weak continuity of the SCP operator over a specific space of probability measures.\footnote{We follow the probabilist convention of writing weak convergence in place of the weak* convergence of functional analysis.} 

More practically, the theorems support applications where the parameters of interest are integer-valued. The regularity conditions such as $P_{\Lambda}^t\in\mathscr{P}_{\Lambda,Q}$ (Assumption \ref{assumption:pushforward}) are required to directly apply the SCP algorithm, but we can instead consider a sequence of TGDs approximating the discrete distribution and the associated sequence of SCP solutions. We can similarly approach scenarios where the input parameter is fully deterministic but unknown, as in \cite{kennedy_ohagan}.

 Let $\cP(\Lambda)$ be the set of all Borel probability measures on $\Lambda$. For a sequence of measures $\{\nu_i\}_{i\in\NN}\subset \cP(\Lambda)$ and $\nu^{\ast}\in \cP(\Lambda)$, we write $\nu_i\stackrel{w}{\to}\nu^{\ast}$ if $\nu_i$ converges weakly to $\nu^{\ast}$. Define the weak sequential closure of $\mathscr{P}_{\Lambda,Q}$ by \begin{align*}
    \overline{\mathscr{P}_{\Lambda,Q}}=\{\nu\in \cP(\Lambda):\text{ there exists }\{\nu_i\}_{i\in\NN} \subset \mathscr{P}_{\Lambda,Q}\text{ such that } \nu_i\stackrel{w}{\to} \nu\}.
\end{align*} This set contains precisely all measures that may be limits of sequences of TGDs. Clearly $\mathscr{P}_{\Lambda,Q}\subseteq \overline{\mathscr{P}_{\Lambda,Q}}$, but we demonstrate in \cref{theorem:characterize:weak:sequential:closure} that $\overline{\mathscr{P}_{\Lambda,Q}}$ also contains any measure of the form $\alpha \nu_1+(1-\alpha)\nu_2$ for $\alpha\in[0,1]$ where $\nu_1\in \mathscr{P}_{\Lambda,Q}$ and $\nu_2$ is a countable mixture of Dirac measures in $\Lambda$.

Our main result, \cref{riesz:theorem}, states that weak convergence of a sequence of TGDs in $\mathscr{P}_{\Lambda,Q}$ to nearly any element of $\overline{\mathscr{P}_{\Lambda,Q}}$ implies weak convergence of the associated SCP solutions. The proof is an application of the Riesz Representation Theorem with disintegration. The argument relies on continuity properties of the surface integrals established in \cite{esip_2025}.  In particular, we consider the ratio $\frac{g_f(q)}{\tilde \rho_{p,\cD}(q)}$ on $\cD$ for any bounded, continuous function $f$, where \[g_f(q)=\int_{Q^{-1}(q)} \frac{ f\cdot \rho_p}{\sqrt{\det J_{Q}J_{Q}^T}} ds,\quad  \tilde \rho_{p,\cD}(q)=\int_{Q^{-1}(q)} \frac{ \rho_p}{\sqrt{\det J_{Q}J_{Q}^T}} ds.\] By Lemma \ref{lemma:inner:integral:equality}, $g_f(q)$ can be seen as the expectation of $f$ with respect to the regular conditional prior measure on the contour $Q^{-1}(q)$.   We show that weak convergence for the SCP solutions $\{\hat P_{\Lambda}^{t,i}\}_{i\in\NN}$ corresponding to $\{P_{\Lambda}^{t,i}\}_{i\in\NN}$ follows if the sequence  $\int_{\cD} \frac{g_f(q)}{\tilde \rho_{p,\cD}(q)} QP_{\Lambda}^{t,i}(q)$ converges to $\int_{\cD} \frac{g_f(q)}{\tilde \rho_{p,\cD}(q)} QP_{\Lambda}^{t,\ast}(q)$ as $i\to\infty$ for all appropriate $f$. We then explicitly construct the limiting SCP solution measure in a few special cases (\cref{theorem:weak:convergence:nice:case,theorem:weak:convergence:point:mass,corollary:weak:convergence:mixture}). 

This section involves two more assumptions, stated below. Assumption \ref{assumption:continuous:surface} is required throughout. Assumption \ref{assumption:surface:integral}, like Assumption \ref{assumption:pushforward}, pertains to the TGD, so as before, we are explicit wherever Assumption \ref{assumption:pushforward} and \ref{assumption:surface:integral} are required. We continue to use Assumptions \ref{assumption:compact}-\ref{assumption:prior}. 

\begin{assumption} \label{assumption:continuous:surface}
Assuming the boundary of $\Lambda$ is given by $\partial \Lambda=\{\lambda:B(\lambda)=0\}$, where $B$ is continuously differentiable on $U_{\Lambda}$, if $\partial \Lambda\cap Q^{-1}(q)$ is non-empty for $q\in\cD$, then it is the set of $\lambda$ solving the system $Q(\lambda)=q$ and $B(\lambda)=0$. For such $q$, $\left(J_Q \; J_B\right)^T$ is of full rank on $\partial \Lambda\cap Q^{-1}(q)$.
\end{assumption}

\begin{assumption} \label{assumption:surface:integral} The TGD $P_{\Lambda}^t$ is such that $QP_{\Lambda}^{t}(\{q\in\cD:\tilde\rho_{p,\cD}(q)=0\})=0$, and there exist infinitely many $k\in\NN$ such that $QP_{\Lambda}^{t}(\{q\in\cD:\tilde\rho_{p,\cD}(q)=1/k\})=0$. 
\end{assumption}

 Assumption \ref{assumption:continuous:surface} was first introduced in \cite{esip_2025} to induce continuity in the SCP solution densities. We use Assumption \ref{assumption:surface:integral}  to ensure that the surface integral $\tilde\rho_{p,\cD}$ has some constraints on its range with respect to the limiting pushforward $Q P_{\Lambda}^{t,\ast}$ for some sequence $\{P_{\Lambda}^{t,i}
 \}_{i\in\NN}\subset \mathscr{P}_{\Lambda,Q}$. We note that for each $i$, we already have $QP_{\Lambda}^{t,i}(\{\tilde\rho_{p,\cD}=0\})=0$. This is because the proof of \cite[Theorem 3.5]{esip_2025}  establishes that $\mu_{\cD}(\{\tilde\rho_{p,\cD}=0\})=0$, and $QP_{\Lambda}^{t,i}\ll \mu_{\cD}$ since  $P_{\Lambda}^{t,i}\in\mathscr{P}_{\Lambda,Q}$. The impact of Assumption \ref{assumption:surface:integral} arises when applied to $QP_{\Lambda}^{t,\ast}$, which may not have a density. The second part of the assumption, i.e., that $QP_{\Lambda}^{t,\ast}(\{\tilde\rho_{p,\cD}=1/k\})=0$ for infinitely many $k$, is a consequence of the proofs employing a  sequence of functions $\{h_k\}_{k\in\NN}$ that approximates the ratio $g_f/\tilde\rho_{p,\cD}$. Each $h_k$ zeroes out the inputs where $\tilde\rho_{p,\cD}>1/k$ to preserve boundedness of this ratio, so this second condition ensures the discontinuity set of each $h_k$ is of null measure, permitting the weak convergence results.

In each of the remaining theorems, $\{P_{\Lambda}^{t,i}\}_{i\in\NN}$ is a sequence of measures in $\mathscr{P}_{\Lambda,Q}$. We assume $P_p$ is a prior on $\Lambda$ with a density $\rho_p = \frac{d P_p}{d\mu_{\Lambda}}$ that is a.e. continuous. The corresponding sequence of SCP solutions is written as $\{\hat P_{\Lambda}^{i}\}_{i\in\NN}$. 

\begin{theorem}[Weak Convergence] \label{riesz:theorem}  Suppose $P_{\Lambda}^{t,i}\stackrel{w}{\to} P_{\Lambda}^{t,\ast}$, where $P_{\Lambda}^{t,\ast}\in\overline{\mathscr{P}_{\Lambda,Q}}$ and satisfies Assumption \ref{assumption:surface:integral}. Then $\hat P_{\Lambda}^{i}\stackrel{w}{\to} \hat\nu$, for some Borel probability measure $\hat \nu$ on $\Lambda$.
\end{theorem}

This  result ensures weak convergence of the SCP solutions without constructing the limiting measure $\hat\nu$. We construct $\hat\nu$  by rewriting the limiting integral $\int_{\cD}\frac{g_f(q)}{\tilde\rho_{p,\cD}(q)}d(QP_{\Lambda}^{t,\ast})(q)$ in the form $\int_{\Lambda} f (\lambda)d\hat\nu(\lambda)$. The difficulty is that in general, $QP_{\Lambda}^{t,\ast}$ is hard to analyze.  In some special cases, it can be clearly related back to the input space. In the first special case, we assume the limiting distribution is itself an element of $\mathscr{P}_{\Lambda,Q}$, and claim that the SCP solution in the limit is precisely the SCP solution applied directly to the limiting TGD.

\begin{theorem}[Constructive Weak Convergence] \label{theorem:weak:convergence:nice:case}  Suppose $P_{\Lambda}^{t,i}\stackrel{w}{\to} P_{\Lambda}^{t,\ast}$, where $P_{\Lambda}^{t,\ast}\in\mathscr{P}_{\Lambda,Q}$ and satisfies the second part of Assumption \ref{assumption:surface:integral}.  Then $\hat P_{\Lambda}^{i}\stackrel{w}{\to} \hat P_{\Lambda}^{\ast}$, where $\hat P_{\Lambda}^{\ast}$ is the solution of the SCP to $P_{\Lambda}^{t,\ast}$.
\end{theorem}

The next special case relaxes the requirement that the limiting TGD is in $\mathscr{P}_{\Lambda,Q}$ by considering a Dirac delta measure at some $\lambda^{\ast}$, where $\lambda^{\ast}$ is assumed to belong to a set of prior probability 1. We find that the solution is a measure concentrating on the contour corresponding to $Q(\lambda^{\ast})$.  We emphasize that for weak convergence, \cref{riesz:theorem}  only requires that $\lambda^{\ast}$ satisfies $\tilde\rho_{p,\cD}(\lambda^{\ast})>0$ to give Assumption \ref{assumption:surface:integral}. In \cref{theorem:weak:convergence:point:mass}, the restriction to a set of prior probability 1 is required to explicitly construct the limiting SCP measure. This is due to the construction relying on an almost sure representation of $\frac{g_f}{\tilde\rho_{p,\cD}}$ from \cref{lemma:inner:integral:equality}.

\begin{theorem}[Constructive Weak Convergence with Dirac Measure] \label{theorem:weak:convergence:point:mass} There exists  $E\in\cB_{\cD}$ with $P_p(Q^{-1}(E))=1$ with the following property. Let $\lambda^{\ast}\in Q^{-1}(E)$ and $q^{\ast}=Q(\lambda^{\ast})$. Suppose $P_{\Lambda}^{t,i}\stackrel{w}{\to} \delta_{\lambda^{\ast}}$, where $\delta_x$ denotes the Dirac delta measure at $x\in\Lambda$. Then $\hat P_{\Lambda}^i \stackrel{w}{\to} \hat P_{\Lambda}^{\ast}$, where $\hat P_{\Lambda}^{\ast}$ is the regular conditional probability $P_p(\,\cdot\,| q^{\ast})$.
\end{theorem}

\Cref{theorem:weak:convergence:point:mass} immediately extends to finite mixtures of Dirac measures, whose supporting points are chosen from this prior probability 1 set, and consequently, any mixture of such finite mixtures with suitable measures in $\mathscr{P}_{\Lambda,Q}$. However, we are unable to extend the technique to countable mixtures due to the requirements of Assumption \ref{assumption:surface:integral}.

\begin{corollary}[Constructive Weak Convergence of Mixtures] \label{corollary:weak:convergence:mixture} There exists $E\in\cB_{\cD}$ with $P_p(Q^{-1}(E))=1$ with the following property. Let $\alpha\in[0,1]$, $\{\omega_i\}_{i=1}^r\subset[0,1]$ with $\sum_{i=1}^r \omega_i=1$, and $\{\lambda_i\}_{i=1}^r\subset Q^{-1}(E)$. Suppose $P_{\Lambda}^{t,i}\stackrel{w}{\to} \nu$ where $\nu=\alpha\nu_1+(1-\alpha)\nu_2$, $\nu_1\in\mathscr{P}_{\Lambda,Q}$ and $\nu_2=\sum_{i=1}^r\omega_i\delta_{\lambda_i}$. Assume further $\nu_1$ satisfies the second part of Assumption \ref{assumption:surface:integral}. Then $\hat P_{\Lambda}^i \stackrel{w}{\to} \hat\nu$ where $\hat\nu:=\alpha \hat\nu_1 + (1-\alpha) \sum_{=1}^r\omega_i P_{p}(\,\cdot\,  | q_i)$,  $\hat\nu_1$ is the SCP solution to $\nu_1$ and $P_p(\,\cdot\,  | q_i)$  is the regular conditional probability for $P_p$ given $q_i=Q(\lambda_i)$. \end{corollary}

In conclusion, we show for a broad class of weakly convergent TGD sequences, the associated SCP solution sequence also weakly converges. We  construct
the limiting SCP solution measure in two cases, one in which the limiting TGD belongs to $\mathscr{P}_{\Lambda,Q}$, and another in which it is a Dirac measure. Consequently, we were able to construct the limiting measure with any finite mixtures of these scenarios. These weak convergence results require assumptions on the  zero set of the surface integral $\tilde\rho_{p,\cD}$. Moreover,  in the constructive Dirac measure scenarios, we must ensure the regular conditional measure admits a surface integral representation, which  only holds almost surely with respect to the prior.

\section{Examples}
\label{section:example}

We present two examples  to illustrate our theoretical continuity results as well as some of the practical strategies involved with using the SCP.\footnote{The associated code can be found at \url{https://github.com/akprasadan/escp-continuity}.}

\subsection{Gaussian Mixture}

\begin{figure}[t]
    \centering
    \includegraphics[width=0.9\textwidth]{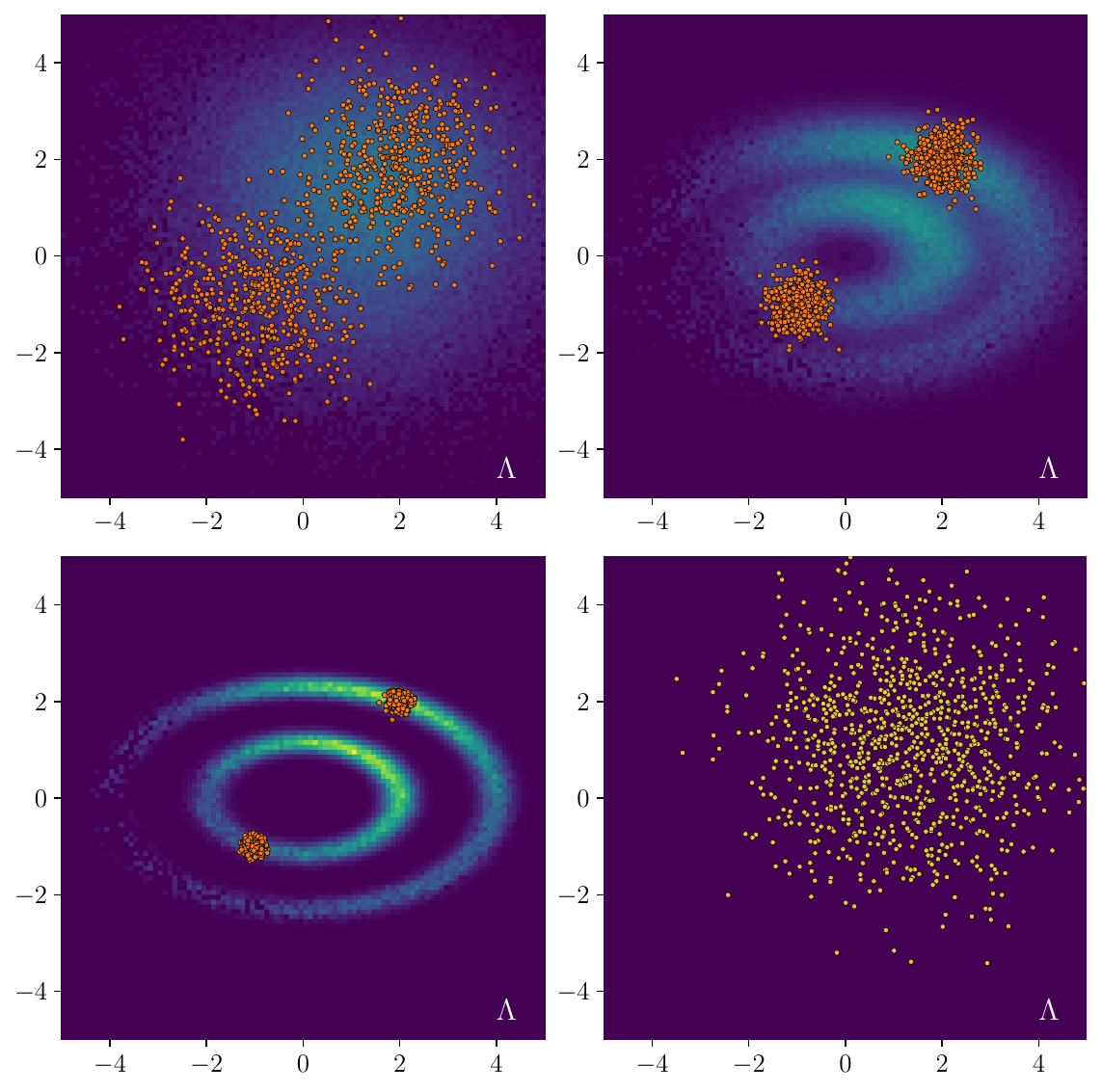} % Adjust width as needed
    \caption{The top row and bottom left plot show a heatmap of the estimated SCP solutions $\hat P_{\Lambda}^i$ on $\Lambda=[-5,5]\times[-5,5]$ with output map $Q(\lambda_1,\lambda_2)=\lambda_1^2+3\lambda_2^2$. Over each heatmap, we overlay as points a subsample of the respective choice of TGD $P_{\Lambda}^{t,i}$. The bottom-right plot shows a subsample of the generated prior datapoints in $\Lambda$.}
    \label{fig:final_example}
\end{figure}

Our first simulation is illustrated in \cref{fig:final_example}. We set $\Lambda=[-5,5]\times[-5,5]$, $Q(\lambda_1,\lambda_2)=\lambda_1^2+3\lambda_2^2$, and consider a sequence of TGDs converging to an equally weighted mixture of point masses at $\mu_A=(-1,-1)$ and $\mu_B=(2,2)$. In particular, we take a mixture of truncated (on $\Lambda$) Gaussians centered at these points, with variance decaying from $1$ to $0.01$ and then $0.01$. The prior (in the bottom-right plot) is a truncated Gaussian centered at $(1.25,1.25)$. We use 100,000 samples for the simulated observations, 100,000 for the prior used in the importance sampling algorithm, partition $\cD$ into a grid of 100 bins, and estimate $\hat P_{\Lambda}$ on a homogeneous rectangular tiling $A_{ij}$ where $1\le i,j\le 100$. We then plot the resulting estimate as a heatmap, and overlay a subset of the synthetic TGD points for reference.

Since $Q$ evaluates on $\mu_A$ to $\mu_B$ to $4$ and $16$, respectively, the associated sequence of SCP solutions concentrate on the pair of ellipses $Q^{-1}(4)$ and $Q^{-1}(16)$ in $\Lambda$. Note that the SCP solution does not directly use the TGD points (plotted in orange), but only their evaluations by $Q$. Consequently, it is expected that the heatmap cannot concentrate at $\mu_A$ and $\mu_B$ due to the non-injectivity of $Q$. The prior, which is closer to $\mu_B$ than $\mu_A$, causes the ellipses to place more mass in the top-right portion of the contours. This reflects that the regular conditional measures $P_p(\,\cdot\,|q)$, produced by disintegrating $P_p$ by $Q$, places more weight near $\mu_B$.

\subsection{A Concrete Application}

\begin{figure}[t]
    \centering
    \includegraphics[width=0.9\textwidth]{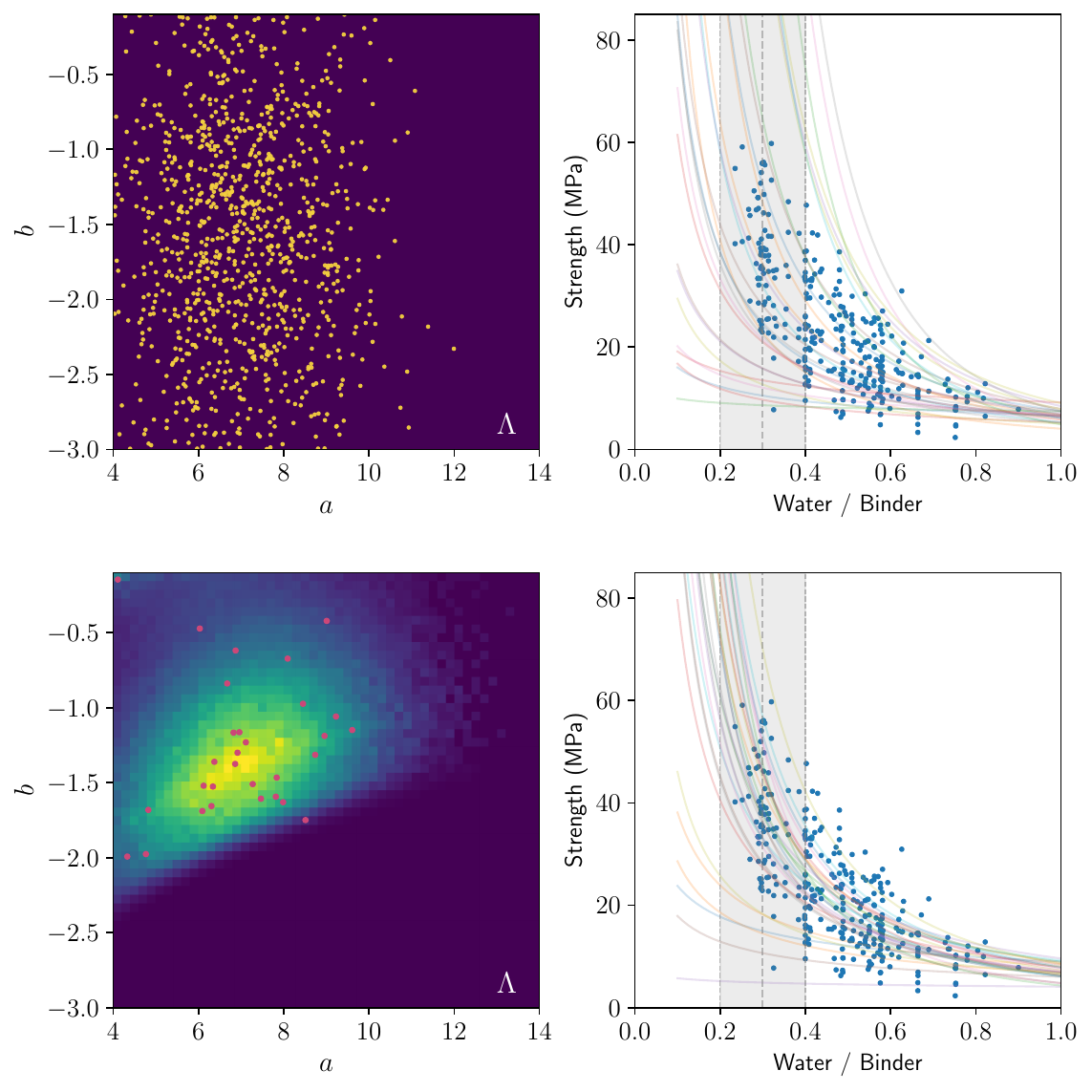} % Adjust width as needed
    \caption{Top left: Scatterplot of points drawn from (truncated) Gaussian prior on $\Lambda=\{(a,b):a\in[4,14], b\in[-3,-0.1]\}$. A point $(a,b)\in \Lambda$ corresponds to a power law $y=aR^b$ for $R>0$. Top right: Compressive strength (megapascals) versus water to binder ratio $R$, with power law fits sampled from prior. Bottom left: SCP solution for $\Lambda$ using Gaussian prior from top left panel and strength values in the vertical strip $R\in[0.2,0.4]$ displayed in the top right panel. Bottom right: Same plot as top right except we plot power laws sampled from the SCP solution, whose corresponding points are also indicated on the bottom left panel.}
    \label{fig:example_three}
\end{figure}

We next apply the SCP algorithm to a popular dataset of concrete attributes from \cite{yeh_1998}.  It has been long known that the compressive strength of concrete is inversely related to the ratio of water to cement used in its mix, e.g., Abrams' Law or its variants. The goal is to solve an inverse problem to estimate this functional form, using strength data.  We  use the water to binder ratio in place of water to cement, where binder refers to the combination of cement, fly ash, and slag added to the mix. \cite{yeh_1998} points out that materials other than cement and water impact strength, and in particular, demonstrate that the water to binder ratio yields better predictive accuracy.

We use the power law $Q(a, b)=aR^b$ for strength, where $R$ is fixed at 0.3, and aim to estimate a distribution over $(a,b)$. By fixing $R$, we have a simple 2-dimensional problem, but note that pairs of $(a,b)$ can then be used to generate predictions for all $R$. In practice, to ensure sufficient data is available, we subset all strength values in the range $R\in[0.2,0.4]$. Moreover, it is known that the strength curve depends on the age of concrete, so we restrict to observations between $0$ and $25$ days old. This yields a subset of samples of size 110. In a followup experiment, we instead consider $25$ to $50$ days, and in doing so, show  the continuity of the solution.

We take as domain $\Lambda=\{(a,b):a\in[4, 14], b\in[-3,-0.1]\}$ and set as prior the  independent coupling $\mathcal{N}(7, 1.5^2)\times \mathcal{N}(-1.5, 1^2)$, truncated to $\Lambda$.  Choosing $\Lambda$ and its prior are key components of SCP estimation. It is important to choose a domain that is broad enough to encompass physical reality, or one informed by real-world expertise. We later demonstrate that our choices are quite permissive, including power laws that are physically unrealistic.  We set $\cD$ as the strength in megapascals. To practically run the SCP algorithm, we need much larger sample sizes, so we use bootstrapping to resample 100,000 observations from the subset of strength data with $R\in[0.2,0.4]$ and age less than $25$ days. We also Gaussian noise with mean $0$, standard deviation $5$ to each resample. The additional noise addresses the low resolution of the data, since the strength is typically only reported in whole digit units of megapascals.

To perform the SCP estimation, we  use a $50\times50$ grid over $\Lambda$, partition $\cD$ into $100$ equally-sized bins, and draw $250,000$ samples from the prior. We show the results in \cref{fig:example_three}. In the top left panel, we display a sub-sample from the truncated Gaussian prior on $\Lambda$. In the upper right, we show the strength and $R$ values for concrete newer than 25 days. They greyed-out region $R\in[0.2,0.4]$ is the actual subset of data we use in the SCP algorithm. We then draw 30 points from the prior samples, and plot their corresponding power law $y=aR^b$ for all $R$. It is evident that our choice of $\Lambda$ and prior are generous in permitting extreme power laws. For example, it is unlikely that the strength shows \textit{no} change between a mix that is 20\% water relative to cement versus one that is equal water and equal cement, as some of the power laws imply.

In the bottom left panel, we show the SCP solution derived from this prior and data subset. For the bottom right panel, we show the same scatterplot except now we show 30 power laws sampled from the SCP solution. The band of power laws from the prior is compressed into a narrower one better fitting the data, aside from a few outliers (corresponding to outlier draws from the SCP solution).

The right column also motivates the use of $R$ closer to $0$ than $1$. All the power laws flatten out as $R\to 1$, and it is much harder to distinguish sensible parameters $(a,b)$ from the strength data in that region. By contrast, for $R\approx 0.3$, different values of $(a,b)$ correspond to dramatically different strength values. Indeed, fitting the SCP algorithm with $R\in[0.6,0.8]$, for example, gives numerically unstable and physically unrealistic results.

\begin{figure}[t]
    \centering
    \includegraphics[width=0.9\textwidth]{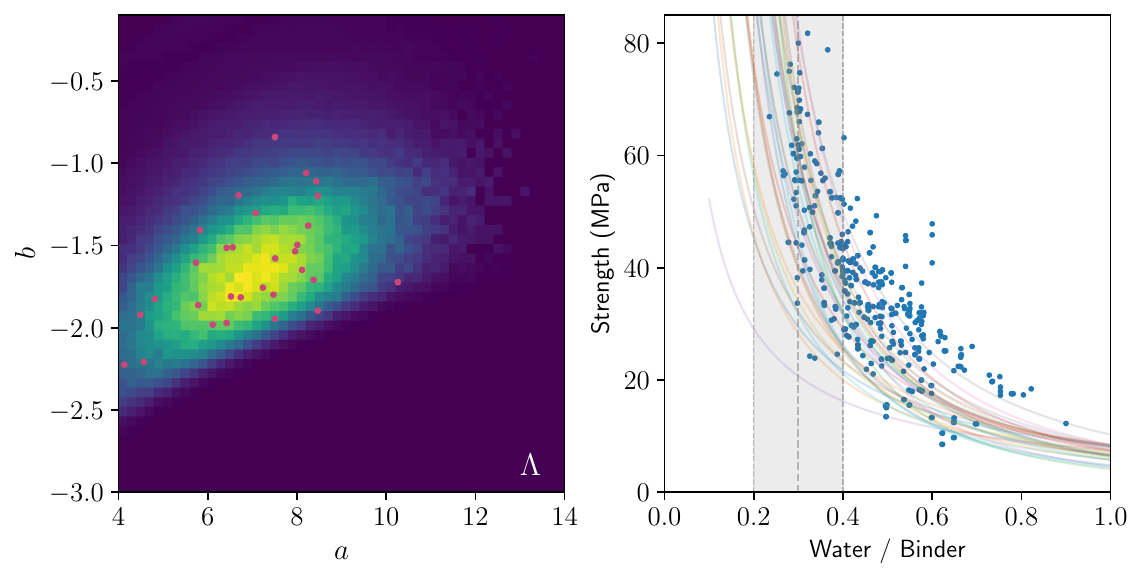} % Adjust width as needed
    \caption{We display the equivalent of the bottom left and bottom right panels of \cref{fig:example_three}, except now we use concrete data that is aged between 25 and 50 days. We plot those strength and water to binder ratio $R$ on the right. The vertical segment $R\in[0.2,0.4]$ depicts the data used in the SCP estimation, whose solution is plotted on the left. We then pick 30 points from this SCP solution (shown as points on the left plot) and graph the corresponding power laws on the right.}
    \label{fig:example_three_part_two}
\end{figure}

To relate these results to the continuity theorems, we repeat this experiment with all else held equal other than the age of the concrete, now set to $[25, 50]$ days, yielding 141 observations before bootstrapping. We expect to see a shift in the $(a,b)$ distribution, albeit not a dramatic one; for example, we expect a power law with negative slope. This prior expectation is also informed by the age-dependent power law formulations of \cite{YEH20061865}. Indeed, \cref{fig:example_three_part_two} shows a tightening in the variance along the $b$-axis but otherwise a similar shape. The band of power laws sampled from the estimate are again much tighter than those from the prior (shown in the top right panel of \cref{fig:example_three}). We remark that the fit seems worse when extrapolated to larger $R$, which is likely a consequence of treating $R$ as fixed rather than solving a 3-dimensional SCP.

\section{Discussion}
\label{section:discussion}

We have established essential theoretical properties of the SCP algorithm developed in \cite{esip_2025}, specifically, continuity of the solution operator over a broad space of probability measures. The disintegration techniques of the SCP have already been used in numerous scientific applications, ranging from coastal hydrodynamics to the contamination of groundwater acquifers \cite{manning_2015, groundwater_2015}. The stability results ensure those applications are robust to small shifts in the trial-generating distribution, that could be caused by changing experimental conditions or typical temporal variation. The continuity results also let us analyze sequences of trial-generating distributions, useful if, say, the experimental conditions become increasingly refined over time. This also suggests connections to fixed point estimation, where the sequence converges to a point mass, as we elaborate below.

Next,  while we have investigated continuity in the TGD, future work may also quantify continuity in the map $Q$ itself. For example, the differential equations  defining a physical model are often just an approximation; thus, we may ask how the approximation error in $Q$ leads to variation in the SCP solution. Alternatively the form $Q$ could vary in time even while the TGD remains the same.

On a more technical level, it remains unclear whether a broader class of limiting TGDs are permitted in our weak convergence results, as well as more that admit explicit constructions. For example, we imposed restrictions on the null set of the surface integral $\tilde \rho_{p,\cD}$ that may in principle be hard to check. Moreover, the Dirac measure scenarios had prior-specific restrictions in order to permit an explicit representation; it remains unclear whether this restriction can be dropped altogether or made less sensitive to the prior.

Perhaps the most interesting future work from this analysis is applying the SCP to the classic Bayesian point estimation problem; that is, there is a fixed, unknown parameter on which we place a prior that captures our uncertainty. In particular, we note that the generation of random variables can be achieved using the well-known inverse transform technique. If we take $Q$ to be the inverse CDF for a particular $\lambda\in\Lambda$ and $u\in[0,1]$ drawn uniformly randomly, then solving the SCP resembles the well-known empirical Bayes set-up of \cite{efron_2014}. Taking the distribution over $\Lambda$ to a point mass, this turns into a fixed point parameter estimation problem.  However, this theory requires modification since the measures on $\Lambda\times[0,1]$ we estimate should preserve the uniform the marginal over $[0,1]$. Therefore, a constrained SCP analogue is needed, as is developed in \cite{csip_2025}. A full development of this connection is ongoing.

\appendix

\section{Proofs for section \ref{section:total:variation}}

    \begin{proof}[Proof of \cref{theorem:stability}]
    In the following calculations, we  substitute \cref{eSIP:density:solution}, perform a change of measure using $\rho_p  = \frac{dP_p}{d\mu_{\Lambda}}$, apply the disintegration theorem (see, e.g., \cite[Theorem 2.4]{esip_2025})  to the prior measure $P_p$, observe $Q(\lambda)=q$ for $\lambda\in Q^{-1}(q)$, note the inner integral integrates to 1 on $Q^{-1}(q)\cap\Lambda$ since $P_p(\,\cdot\, | q)$ is regular conditional measure, perform another change of measure using $\tilde\rho_{p,\cD}=\frac{d QP_p}{d\mu_{\cD}}$, and finally recognize the total variation distance of $p_{\cD}$ and $r_{\cD}$. This gives
         \begin{align*}
            2\cdot d_{\TV}(\hat p_{\Lambda},\hat r_{\Lambda})  &= \int_{\Lambda} |\hat p_{\Lambda}(\lambda)-\hat r_{\Lambda} (\lambda)|d\mu_{\Lambda}(\lambda)\\
            &=\int_{\Lambda} \frac{\rho_p(\lambda)}{\tilde\rho_{p,\cD}(Q(\lambda))}|p_{\cD}(Q(\lambda))-r_{\cD}(Q(\lambda))|d\mu_{\Lambda}(\lambda) \\
            &=\int_{\Lambda} \frac{1}{\tilde\rho_{p,\cD}(Q(\lambda))}|p_{\cD}(Q(\lambda))-r_{\cD}(Q(\lambda))|dP_{p}(\lambda) \\
            &= \int_{\cD} \left[\int_{Q^{-1}(q)\cap\Lambda} \frac{1}{\tilde\rho_{p,\cD}(Q(\lambda))}|p_{\cD}(Q(\lambda))-r_{\cD}(Q(\lambda))| dP_{p}(\lambda| q)(\lambda) \right]d QP_p(q) \\
            &=  \int_{\cD} \frac{|p_{\cD}(q)-r_{\cD}(q)|}{\tilde\rho_{p,\cD}(q)}\left[\int_{Q^{-1}(q)\cap\Lambda}  dP_{p}(\lambda| q) \right]d QP_p(q) \\
            &= \int_{\cD} \frac{|p_{\cD}(q)-r_{\cD}(q)|}{\tilde\rho_{p,\cD}(q)}\; d QP_p(q) = \int_{\cD} |p_{\cD}(q)-r_{\cD}(q)| d\mu_{\cD}(q) \\
            &= 2\cdot d_{\TV}(p_{\cD}, r_{\cD}).
        \end{align*} By a data-processing inequality \cite[Theorem 5.2]{devroye2001combinatorial}, \[d_{\TV}(p_{\cD},r_{\cD})  =d_{\TV}\left( QP_{\Lambda}^t,QR_{\Lambda}^t\right)\le d_{\TV}\left( P_{\Lambda}^t,R_{\Lambda}^t\right).\] 
    \end{proof}

\section{Proofs for  section \ref{section:weak:continuity}}

We first prove the claimed characterization of $\overline{\mathscr{P}_{\Lambda,Q}}$ as containing countable mixtures of Dirac measures and elements of $\mathscr{P}_{\Lambda,Q}$. We start with the simple Dirac measure case by adapting \cite[Example 8.1.4]{Bogachev2007}.

\begin{lemma} \label{lemma:existence:converging:to:dirac} Pick $\lambda^{\ast}\in\Lambda$. Then $\delta_{\lambda^{\ast}}\in \overline{\mathscr{P}_{\Lambda,Q}}$.
\end{lemma}
    \begin{proof}
        Let $\rho$ be any density on $\Lambda$ with respect to $\mu_{\Lambda}$, and assume $\rho$ is 0 outside of $\Lambda$.  Define $\rho_n(\lambda) = n\cdot \rho(n(\lambda-\lambda^{\ast}))$ and let $\nu_n$ be the associated probability measure on $\Lambda$. As argued earlier, each $\nu_n\in \mathscr{P}_{\Lambda,Q}$ since it admits a density. Pick any bounded, continuous function $f$ on $\Lambda$. Then \begin{align*}
            \int f(\lambda)d\nu_n(\lambda) &= \int f(\lambda) \rho_n(\lambda) d\lambda = \int f(\lambda)\cdot n\rho(n(\lambda-\lambda^{\ast})) d\lambda \\
            &= \int f(\lambda^{\ast}+t/n) \rho(t) dt \to f(\lambda^{\ast}) = \int f(\lambda) d\delta_{\lambda^{\ast}}.
        \end{align*} Hence $\nu_n\stackrel{w}{\to}\delta_{\lambda^{\ast}}$, proving $\delta_{\lambda^{\ast}}\in \overline{\mathscr{P}_{\Lambda,Q}}$.
    \end{proof}

\begin{theorem} \label{theorem:characterize:weak:sequential:closure} Let $\cK\subseteq\NN$, $\{\lambda_i\}_{i\in\cK}\subset\Lambda$, and $\{p_i\}_{i\in\cK}$ such that each $p_i\ge 0$ and $\sum_{i\in\cK} p_i=1$. Pick $\alpha\in[0.1]$. Set $\nu_1\in\mathscr{P}_{\Lambda,Q}$ and $\nu_2=\sum_{i\in\cK} p_i\delta_{\lambda_i}$ and define $\nu=\alpha\nu_1+(1-\alpha)\nu_2$. Then $\nu\in \overline{\mathscr{P}_{\Lambda,Q}}$.
\end{theorem}
    \begin{proof} For each $i\in\cK$, let $\{\eta_{j}^{(i)}\}_{j\in\NN}\subset \mathscr{P}_{\Lambda,Q}$ such that $\lim_j \eta_{j}^{(i)}=\delta_{\lambda_i}$ in the weak sense; these sequences exist and have a density by the proof of \cref{lemma:existence:converging:to:dirac}. Define the sequence of mixtures $\{\eta_j\}_{j\in\NN}$ by $\eta_j = \alpha\nu_1+(1-\alpha)\sum_{i\in\cK}p_i\eta_j^{(i)}$. Each $\eta_j\in\mathscr{P}_{\Lambda,Q}$ because it is a convex combination of measures with a density, hence has a density itself. Pick a bounded, continuous $f$ on $\Lambda$. Then \begin{align*}
        \int f(\lambda) d\eta_j(\lambda) &= \alpha \int f(\lambda) d\nu_1(\lambda) + (1-\alpha)\sum_{i\in\cK} p_i\int f(\lambda) d\eta_j^{(i)}(\lambda) \\
        &\to \alpha \int f(\lambda) d\nu_1(\lambda) + (1-\alpha)\sum_{i\in\cK} p_i\int f(\lambda) d\delta_{\lambda_i}(\lambda) \\
        &= \alpha \int f(\lambda) d\nu_1(\lambda) + (1-\alpha) \int f(\lambda) d\nu_2(\lambda) = \int f(\lambda) d\nu(\lambda).
    \end{align*} Hence $\eta_j \stackrel{w}{\to}\nu$, proving $\nu\in \overline{\mathscr{P}_{\Lambda,Q}}$.        
    \end{proof}

We proceed to the main proofs of weak convergence of the SCP solution operator. The first theorem is the general existence result, and we follow it up with three constructive results.

    \begin{proof}[Proof of \cref{riesz:theorem}]
        Let $\rho_{\cD,i}=\frac{d(Q P_{\Lambda}^{t,i})}{d\mu_{\cD}}$, which exists by definition of $\mathscr{P}_{\Lambda,Q}$.  Let $\tilde\rho_{p,\cD} =
\frac{d (Q P_p)}{d\mu_{\cD}}$ and  $
\tilde\rho_{\cD} = \frac{d(Q\mu_{\Lambda})}{d\mu_{\cD}}$ be the surface integrals defined in \cite[Theorem 3.5 and Theorem E.2]{esip_2025}. Since the latter theorem establishes that $\mu_{\cD}(\{\tilde\rho_{\cD}=0\})=0$, $Q\mu_{\Lambda}$ and $\mu_{\cD}$ are absolutely continuous with respect to each other. Hence,  \begin{equation} \label{eq:radon:nikodym:chain}
    \frac{\rho_{\cD,i}}{\tilde\rho_{\cD}}=\frac{d (QP_{\Lambda}^{t,i})}{d\mu_{\cD}}\cdot \frac{d\mu_{\cD}}{d (Q\mu_{\Lambda})} = \frac{d (QP_{\Lambda}^{t,i})}{d (Q\mu_{\Lambda})}.
\end{equation} Let $\hat p_{\Lambda}^i$ be the density of the SCP solution for $P_{\Lambda}^{t,i}$, which from \cref{eSIP:density:solution} is \[\hat p_{\Lambda}^i(\lambda)=\frac{d\hat P_{\Lambda}^i}{d\mu_{\Lambda}}(\lambda) =\frac{\rho_{\cD,i}(Q(\lambda))\rho_p(\lambda)}{\tilde\rho_{p,\cD}(Q(\lambda))}.\] Define for any bounded, continuous function $f$ on $\Lambda$ the function $g_f$ on $\cD$ by \[g_f(q)= \int_{Q^{-1}(q)} \frac{ f\cdot \rho_p}{\sqrt{\det J_{Q}J_{Q}^T}} ds.\]

 Observe $QP_{\Lambda}^{t,i}\stackrel{w}{\to}QP_{\Lambda}^{\ast}$ by   \cite[Theorem 29.2]{billingsley1995probability}. Pick any bounded, continuous $f$ on $\Lambda$.   We apply \cref{eSIP:density:solution}, disintegrate $\mu_{\Lambda}$, apply \cref{lemma:inner:integral:equality}, and use \cref{eq:radon:nikodym:chain}. This gives: \begin{align}
        \int_{\Lambda} f(\lambda)d\hat P_{\Lambda}^i(\lambda) &= \int_{\Lambda}\frac{f(\lambda)\rho_{\cD,i}(Q(\lambda))\rho_p(\lambda)}{\tilde\rho_{p,\cD}(Q(\lambda))} d\mu_{\Lambda}(\lambda)\notag  \\
        &= \int_{\cD} \frac{\rho_{\cD,i}(q)}{\tilde\rho_{p,\cD}(q)}\left[ \int_{Q^{-1}(q)} f(\lambda)\rho_p(\lambda) d\mu_N(\lambda| q)\right] d (Q\mu_{\Lambda})(q)\notag  \\
        &=\int_{\cD} \frac{\rho_{\cD,i}(q)}{\tilde\rho_{p,\cD}(q)\tilde \rho_{\cD}(q)} \left[ \int_{Q^{-1}(q)} \frac{ f\cdot \rho_p}{\sqrt{\det J_{Q}J_{Q}^T}} ds\right] d (Q\mu_{\Lambda})(q) \notag  \\
        &= \int_{\cD}\frac{1}{\tilde \rho_{p,\cD}(q)} \left[ \int_{Q^{-1}(q)} \frac{ f\cdot \rho_p}{\sqrt{\det J_{Q}J_{Q}^T}} ds\right] d (Q P_{\Lambda}^{t,i})(q) \notag \\
        &= \int_{\cD} \frac{g_f(q)}{\tilde \rho_{p,\cD}(q)}d (Q P_{\Lambda}^{t,i})(q). \notag
            \end{align} 
            
    We argue that
\begin{align*}
     \int_{\cD} \frac{g_f(q)}{\tilde \rho_{p,\cD}(q)}d (Q P_{\Lambda}^{t,i})(q) \to \int_{\cD} \frac{g_f(q)}{\tilde \rho_{p,\cD}(q)}d (Q P_{\Lambda}^{t,\ast})(q). 
    \end{align*} Both $g_f$ and $\tilde\rho_{p,\cD}$ are bounded and continuous on $\cD$ by  \cite[Theorem 3.8]{esip_2025}. Pick $\varepsilon>0$. Define  the sequence $\{h_k\}_{k\in\NN}$ of functions on $\cD$ by \[h_k(q)=\frac{g_f(q)}{\tilde\rho_{p,\cD}(q)}\cdot \chi_{\{\tilde\rho_{p,\cD}>1/k\}}(q).\] Observe \begin{align}
        \MoveEqLeft\left|\int_{\cD} \frac{g_f}{\tilde \rho_{p,\cD}}d (Q P_{\Lambda}^{t,i}) -   \int_{\cD} \frac{g_f(q)}{\tilde \rho_{p,\cD}}d (Q P_{\Lambda}^{t,\ast})\right| \notag \\ &\le \left|\int_{\cD} h_{k}\; d(QP_{\Lambda}^{t,\ast})-\int_{\cD} \frac{g_f}{\tilde \rho_{p,\cD}}d (Q P_{\Lambda}^{t,\ast})\right| +  \left|\int_{\cD} \frac{g_f}{\tilde \rho_{p,\cD}}d (Q P_{\Lambda}^{t,i})  - \int_{\cD} h_{k} \; d(QP_{\Lambda}^{t,i})\right| \notag \\
        &\qquad\qquad +\left|\int_{\cD} h_{k} \;d(QP_{\Lambda}^{t,i})-\int_{\cD} h_{k} \;d(QP_{\Lambda}^{t,\ast})\right|. \label{eq:triple:triangle:inequality}
    \end{align} We now choose $k$ such that for all sufficiently large $i$, each of these three terms is smaller than  $\varepsilon/3$.

    Choose $C>0$ such that $|f|\le C$ on $\Lambda$. Since $QP_{\Lambda}^{t,\ast}(\{\tilde\rho_{p,\cD}=0\})=0$,  integrals over the set $\{0<\tilde\rho_{p,\cD}\le 1/k\}$ are equal to those over $\{0\le \tilde\rho_{p,\cD}\le 1/k\}$ with respect to $QP_{\Lambda}^{t,\ast}$. Thus, \begin{align*}
    \left|\int_{\cD} h_k \;d(QP_{\Lambda}^{t,\ast}) -\int_{\cD} \frac{g_f}{\tilde \rho_{p,\cD}}d (Q P_{\Lambda}^{t,\ast})\right| 
    &= \left|\int_{\{0<\tilde\rho_{p,\cD}\le    1/k\}}\frac{g_f}{\tilde \rho_{p,\cD}}d (Q P_{\Lambda}^{t,\ast})\right| \\
        &\le \left|\int_{\{0<\tilde\rho_{p,\cD}\le    1/k\}}\frac{\int_{Q^{-1}(q)} \frac{f\rho_p}{\sqrt{\det J_Q J_Q^T}}  ds}{\int_{Q^{-1}(q)} \frac{\rho_p}{\sqrt{\det J_Q J_Q^T}}  ds}d (Q P_{\Lambda}^{t,\ast})\right| \\
    &\le \left|\int_{\{0<\tilde\rho_{p,\cD}\le    1/k\}}\frac{C\cdot\int_{Q^{-1}(q)} \frac{\rho_p}{\sqrt{\det J_Q J_Q^T}}  ds}{\int_{Q^{-1}(q)} \frac{\rho_p}{\sqrt{\det J_Q J_Q^T}}  ds}d (Q P_{\Lambda}^{t,\ast})\right|  \\
    &\le  C\cdot QP_{\Lambda}^{t,\ast}(\{0\le \tilde\rho_{p,\cD}\le    1/k\}) \\
    &\downarrow  \lim_{k\to\infty} C\cdot QP_{\Lambda}^{t,\ast}(\{\tilde\rho_{p,\cD}=0\}) = 0.
\end{align*} Hence for some $M_1\in\NN$, if $k\ge M_1$, then \begin{equation} \label{eq:h:g:ast:bound}
     \left|\int_{\cD} h_k \;d(QP_{\Lambda}^{t,\ast}) -\int_{\cD} \frac{g_f}{\tilde \rho_{p,\cD}}d (Q P_{\Lambda}^{t,\ast})\right|\le \varepsilon/3.
\end{equation}
    
We perform a similar argument with $QP_{\Lambda}^{t,i}$. Recall $QP_{\Lambda}^{t,i}(\{\tilde\rho_{p,\cD}=0\})=0$ for each $i$ since $QP_{\Lambda}^{t,i}\ll \mu_{\cD}$. This zero set can thus be discarded from integrals in $QP_{\Lambda}^{t,i}$ as before. Hence  \begin{align*}
        \MoveEqLeft \limsup_i \left|\int_{\cD} h_k \;d(QP_{\Lambda}^{t,i}) -\int_{\cD} \frac{g_f}{\tilde \rho_{p,\cD}}d (Q P_{\Lambda}^{t,i})\right|\\ &= \limsup_i\left|\int_{\{0<\tilde\rho_{p,\cD}\le    1/k\}}\frac{g_f}{\tilde \rho_{p,\cD}}d (Q P_{\Lambda}^{t,i})\right| \\
        &\le \limsup_i\left|\int_{\{0<\tilde\rho_{p,\cD}\le    1/k\}}\frac{\int_{Q^{-1}(q)} \frac{f\rho_p}{\sqrt{\det J_Q J_Q^T}}  ds}{\int_{Q^{-1}(q)} \frac{\rho_p}{\sqrt{\det J_Q J_Q^T}}  ds}d (Q P_{\Lambda}^{t,i})\right| \\
    &\le \limsup_i\left|\int_{\{0<\tilde\rho_{p,\cD}\le    1/k\}}\frac{C\cdot\int_{Q^{-1}(q)} \frac{\rho_p}{\sqrt{\det J_Q J_Q^T}}  ds}{\int_{Q^{-1}(q)} \frac{\rho_p}{\sqrt{\det J_Q J_Q^T}}  ds}d (Q P_{\Lambda}^{t,i})\right|  \\
    &\le C\cdot \limsup_i QP_{\Lambda}^{t,i}(\{0\le \tilde\rho_{p,\cD}\le    1/k\})\le C\cdot QP_{\Lambda}^{t,\ast}(\{0\le \tilde\rho_{p,\cD}\le    1/k\}) \\ &\downarrow \lim_{k\to\infty} C\cdot QP_{\Lambda}^{t,\ast}(\{\tilde\rho_{p,\cD}= 0\})= 0.
\end{align*} The first five lines follow by the same reasoning as the $QP_{\Lambda}^{t,\ast}$ case. In the sixth, we apply \cite[Theorem 13.16(iv)]{klenke2020} to the closed (by continuity of $\tilde\rho_{p,\cD}$)  set $\{0\le\tilde\rho_{p,\cD}\le k\}$ along with $QP_{\Lambda}^{t,i}\stackrel{w}{\to} QP_{\Lambda}^{t,\ast}$. Then we use  Assumption \ref{assumption:surface:integral}. Hence, there exists $M_2\in\NN$ such that $k\ge M_2$ implies \begin{equation*} 
    \limsup_i \left|\int_{\cD} h_k\; d(QP_{\Lambda}^{t,i}) -\int_{\cD} \frac{g_f}{\tilde \rho_{p,\cD}}d (Q P_{\Lambda}^{t,i})\right| \le \varepsilon/6.
\end{equation*}

By Assumption \ref{assumption:surface:integral}, there exists some $k^{\ast}\ge\max(M_1,M_2)$ such that $QP_{\Lambda}^{t,\ast}(\{\tilde\rho_{p,\cD}=1/k^{\ast}\})=0$. For such $k^{\ast}$, we have \begin{equation*}  
    \limsup_i \left|\int_{\cD} h_{k^{\ast}}\; d(QP_{\Lambda}^{t,i}) -\int_{\cD} \frac{g_f}{\tilde \rho_{p,\cD}}d (Q P_{\Lambda}^{t,i})\right| \le \varepsilon/6.
\end{equation*}

 By definition of a limsup, there exists $M_3\in\NN$ such that if $i\ge M_3$, then  \begin{equation} \label{eq:h:g:i:bound}
     \left|\int_{\cD} h_{k^{\ast}}\; d(QP_{\Lambda}^{t,i}) -\int_{\cD} \frac{g_f}{\tilde \rho_{p,\cD}}d (Q P_{\Lambda}^{t,i})\right| \le\varepsilon/3.
 \end{equation}

   Then, observe that $h_{k^{\ast}}$ has  discontinuity set   $\{\tilde\rho_{p,\cD}=1/k^{\ast}\}$.  Thus, $h_{k^{\ast}}$ is continuous  $QP_{\Lambda}^{t,\ast}$-a.e.  from Assumption \ref{assumption:surface:integral} and is  bounded. Then, \cite[Theorem 13.16(iii)]{klenke2020} along with $QP_{\Lambda}^{t,i}\stackrel{w}{\to} QP_{\Lambda}^{t,\ast}$ implies $\lim_{i\to\infty}\int_{\cD} h_{k^{\ast}} \;d(QP_{\Lambda}^{t,i})= \int_{\cD} h_{k^{\ast}} \;d(QP_{\Lambda}^{t,\ast})$. So there exists $M_4\in\NN$ such that $i\ge M_4$ implies \begin{equation} \label{eq:h:h:bound}
       \left|\int_{\cD} h_{k^{\ast}} \; d(QP_{\Lambda}^{t,i})-\int_{\cD} h_{k^{\ast}}\; d(QP_{\Lambda}^{t,\ast})\right|< \varepsilon/3.
   \end{equation}
    
    We combine the bounds and return to \cref{eq:triple:triangle:inequality}. Suppose $i\ge \max(M_3,M_4)$ and consider the choice of $k^{\ast}\ge\max(M_1,M_2)$ with $QP_{\Lambda}^{t,\ast}(\{\tilde\rho_{p,\cD}=1/k^{\ast}\})=0$. Using \cref{eq:h:g:ast:bound,eq:h:g:i:bound,eq:h:h:bound}, \begin{align*}
        \MoveEqLeft\left|\int_{\cD} \frac{g_f}{\tilde \rho_{p,\cD}}d (Q P_{\Lambda}^{t,i}) -   \int_{\cD} \frac{g_f(q)}{\tilde \rho_{p,\cD}}d (Q P_{\Lambda}^{t,\ast})\right| \le \varepsilon.
    \end{align*} Thus \begin{equation} \label{eq:weak:continuity:intermediate}
        \int_{\Lambda} f(\lambda)d\hat P_{\Lambda}^i(\lambda)=\int_{\cD} \frac{g_f(q)}{\tilde \rho_{p,\cD}(q)}d (Q P_{\Lambda}^{t,i})(q)\to \int_{\cD} \frac{g_f(q)}{\tilde \rho_{p,\cD}(q)}d (Q P_{\Lambda}^{t,\ast})(q).
    \end{equation}

    Note  that $\Lambda$ is a compact Hausdorff space. Let $\cC_{\Lambda}$ be the set of bounded, continuous functions on $\Lambda$. Define the positive linear functional $\psi \colon \cC_{\Lambda}\to \RR$ by \[\psi(h) = \int_{\cD} \frac{g_h(q)}{\tilde\rho_{p,\cD}(q)} d(Q P_{\Lambda}^{\ast})(q)\] for any $h\in\cC_{\Lambda}$. The operator norm $\|\cdot\|_{\mathrm{op}}$ of $\psi$ satisfies $\|\psi\|_{\mathrm{op}}=\psi(\mathbbm{1})$ according to \cite[Lemma 6.2]{van2003probability}, where $\mathbbm{1}$ is the constant function taking value $1$ on $\Lambda$. Applying \cref{eq:weak:continuity:intermediate} with $f=\mathbbm{1}$, we have \[1=\lim_{i\to\infty}\int_{\Lambda} d\hat P_{\Lambda}^i(\lambda)=\int_{\cD} \frac{g_f(q)}{\tilde \rho_{p,\cD}(q)}d (Q P_{\Lambda}^{t,\ast})(q)=\psi(\mathbbm{1}).\] Hence $\|\psi\|_{\mathrm{op}}=1$. By the Riesz Representation Theorem \cite[Theorem 6.3]{van2003probability} there is a unique Borel probability measure $\hat\nu$ on $\Lambda$ such that \[\psi(h)=\int_{\Lambda} h (\lambda)d\hat\nu(\lambda)\] for all bounded, continuous $h$. In particular, take $h=f$. Thus, we may  write \[\int_{\Lambda} f(\lambda)d\hat P_{\Lambda}^i(\lambda)\to \int_{\cD} \frac{g_f(q)}{\tilde\rho_{p,\cD}(q)}  dQ P_{\Lambda}^{\ast}(q) = \psi(f)=\int_{\Lambda} f (\lambda)d\hat\nu(\lambda).\] Hence, $\hat P_{\Lambda}^i\stackrel{w}{\to}\hat\nu$. 
    \end{proof}

\begin{proof}[Proof of \cref{theorem:weak:convergence:nice:case}] Suppose $P_{\Lambda}^{t,i}\stackrel{w}{\to} P_{\Lambda}^{t,\ast}$. The first part of Assumption \ref{assumption:surface:integral} holds since $P_{\Lambda}^{t,\ast}\in\mathscr{P}_{\Lambda,Q}$ and $\mu_{\cD}(\{\tilde\rho_{p,\cD}=0\})=0$ implies $QP_{\Lambda}^{t,\ast}(\{\tilde\rho_{p,\cD}=0\})=0$. Hence $P_{\Lambda}^{t,\ast}$ satisfies Assumption \ref{assumption:surface:integral}.  Let $f$ be bounded and continuous. Repeating the proof of \cref{riesz:theorem}, we   arrive at \cref{eq:weak:continuity:intermediate}. We show $\int_{\cD} \frac{g_f(q)}{\tilde\rho_{p,\cD}(q)}  d(Q P_{\Lambda}^{t,\ast})(q)=\int_{\Lambda}f(\lambda) d \hat P_{\Lambda}^{\ast}(\lambda)$. Set $\rho_{\cD,\ast}(q) = \frac{d(Q P_{\Lambda}^{t,\ast})}{d\mu_{\cD}}(q)$, which exists since $P_{\Lambda}^{t,\ast}\in\mathscr{P}_{\Lambda,Q}$. By the same reasoning used for \cref{eq:radon:nikodym:chain}, $\frac{\rho_{\cD,\ast}}{\tilde\rho_{\cD}}= \frac{d (QP_{\Lambda}^{t,\ast})}{d (Q\mu_{\Lambda})}.$ Using this fact, \cref{lemma:inner:integral:equality}, the disintegration theorem, and \cref{eSIP:density:solution},  we have \begin{align*}
    \int_{\cD} \frac{g_f(q)}{\tilde\rho_{p,\cD}(q)}  d(Q P_{\Lambda}^{t,\ast})(q) &= \int_{\cD} \frac{1}{\tilde\rho_{p,\cD}(q)} \left[\int_{Q^{-1}(q)} \frac{f\cdot \rho_p}{\sqrt{\det J_QJ_Q^T}}ds\right] d(Q P_{\Lambda}^{t,\ast})(q) \\
    &= \int_{\cD} \frac{\rho_{\cD,\ast}(q)}{\tilde\rho_{p,\cD}(q)} \left[\frac{1}{\tilde \rho_{\cD}(q)} \int_{Q^{-1}(q)} \frac{f\cdot \rho_p}{\sqrt{\det J_QJ_Q^T}}ds\right]  d(Q\mu_{\Lambda})(q) \\
     &= \int_{\cD} \frac{\rho_{\cD,\ast}(q)}{\tilde\rho_{p,\cD}(q)} \left[\int_{Q^{-1}(q)} f(\lambda)\rho_p(\lambda) d\mu_N(\lambda| q)\right] d(Q\mu_{\Lambda})(q)  \\
    &= \int_{\Lambda} \frac{\rho_{\cD,\ast}(Q(\lambda))\rho_p(\lambda)f(\lambda)}{\tilde\rho_{p,\cD}(Q(\lambda))} d\mu_{\Lambda}(\lambda) \\
    &= \int_{\Lambda} f(\lambda) d \hat P_{\Lambda}^{\ast}(\lambda),
\end{align*} Hence by \cref{eq:weak:continuity:intermediate}, $\int_{\Lambda} f(\lambda) d\hat P_{\Lambda}^i(\lambda)\to  \int_{\Lambda} f(\lambda) d \hat P_{\Lambda}^{\ast}(\lambda)$ for all bounded, continuous $f$ as desired.
\end{proof}

    \begin{proof}[Proof of \cref{theorem:weak:convergence:point:mass}]
    Recall the definitions of $g_f$ and $\tilde\rho_{p,\cD}$ from the proof of \cref{riesz:theorem}.  By \cref{lemma:inner:integral:equality},  \begin{equation}
        \frac{g_f(q)}{\tilde \rho_{p,\cD}(q)} = \int_{Q^{-1}(q)} f(\lambda) d P_p(\lambda| q) \label{eq:surface:integral:representation}
    \end{equation} for $QP_p$-a.e.  $q$ (the set of $q$ such that \cref{eq:surface:integral:representation} does not hold has $QP_p$ measure $0$). Now, $P_p\ll \mu_{\Lambda}$, so we have $QP_p\ll Q\mu_{\Lambda}$. Also, \cite[Theorem E.2]{esip_2025} shows $Q\mu_{\Lambda}\ll \mu_{\cD}$. Hence $QP_p\ll \mu_{\cD}$, and since $\mu_{\cD}(\{\tilde\rho_{p,\cD}=0\})=0$, we have $QP_p(\{\tilde\rho_{p,\cD}=0\})=0$. Thus, take \begin{equation} \label{eq:definition:E}
        E= \{q\in\cD:\tilde\rho_{p,\cD}(q)\ne 0 \text{ and }\cref{eq:surface:integral:representation}\text{ holds}\}
    \end{equation} and note $QP_p(E)=1$.

    Pick $\lambda^{\ast}$ from $Q^{-1}(E)$, and note $q^{\ast}$ satisfies \cref{eq:surface:integral:representation}. Observe that $QP_{\Lambda}^{t,\ast}$ satisfies Assumption \ref{assumption:surface:integral}. To see this, note $\tilde\rho_{p,\cD}(q^{\ast})\ne 0$ implies $\delta_{q^{\ast}}(\{\tilde\rho_{p,\cD}=0\})=0$. Then by continuity of $\tilde\rho_{p,\cD}$, for all sufficiently large $k$, $\tilde\rho_{p,\cD}(q^{\ast})\ne 1/k$ which implies $\delta_{q^{\ast}}(\{\tilde\rho_{p,\cD}=1/k\})=0$.  Moreover, $\delta_{\lambda^{\ast}}\in\overline{\mathscr{P}_{\Lambda,Q}}$ as previously discussed. Now suppose $P_{\Lambda}^{t,i}\stackrel{w}{\to}\delta_{\lambda^{\ast}}$. Then we  repeat the argument of \cref{riesz:theorem} leading up to \cref{eq:weak:continuity:intermediate}, followed by \cref{eq:surface:integral:representation}: \begin{align*}
        \int_{\Lambda} f(\lambda) d\hat P_{\Lambda}^i(\lambda) &\to \int_{\cD} \frac{g_f(q)}{\tilde \rho_{p,\cD}(q)}d (Q P_{\Lambda}^{t,\ast})(q) =\frac{g_f(q^{\ast})}{\tilde\rho_{p,\cD}(q^{\ast})} = \int_{Q^{-1}(q^{\ast})} f(\lambda) d P_p(\lambda| q^{\ast}) \\ &=\int_{\Lambda} f(\lambda) d P_p(\lambda| q^{\ast}).
    \end{align*}
    The last equality used the concentration property of regular conditional probability \cite[Definition 1(i)]{chang_pollard}.  Thus, $\hat P_{\Lambda}^i\stackrel{w}{\to}  P_p(\,\cdot\,| q^{\ast})$.
    \end{proof}

    \begin{proof}[Proof of \cref{corollary:weak:convergence:mixture}]
    Note that $\alpha\nu_1+(1-\alpha)\nu_2\in\overline{\mathscr{P}_{\Lambda,Q}}$ by \cref{theorem:characterize:weak:sequential:closure} so the convergence hypothesis is permitted. Let $E$ be defined as in \cref{eq:definition:E}. We now check that Assumption \ref{assumption:surface:integral} is satisfied for $\nu$. Since $Q\nu= \alpha (Q\nu_1)+(1-\alpha)\sum_{i=1}^r \delta_{q_i},$ \begin{align*}
         Q\nu(\{\tilde\rho_{p,\cD}=0\}) &= \alpha (Q\nu_1)(\{\tilde\rho_{p,\cD}=0\}) + (1-\alpha) \sum_{i=1}^r \delta_{q_i}(\{\tilde\rho_{p,\cD}=0\}) = 0.
    \end{align*} This follows since in the proof of \cref{theorem:weak:convergence:nice:case}, $\nu_1$ automatically satisfies the first part of Assumption \ref{assumption:surface:integral}, and the definition of $E$ ensures each $\tilde\rho_{p,\cD}(q_i)\ne 0$. Next, we know by assumption on $\nu_1$ there exists an infinite set $\cK\subseteq\NN$ such that that $(Q\nu_1)(\{\tilde\rho_{p,\cD}=1/k\})=0$. Now for each $i\in\{1,\cdots,r\}$, since $\tilde\rho_{p,\cD}(q_i)\ne 0$, there exists a $k_i$ such that for all $k\ge k_i$, $\tilde\rho_{p,\cD}(q_i)\ne 1/k$. Set $k^{\ast}=\max(k_1,\dots,k_r)$. Then if $k\ge k^{\ast}$, for any $i\in\{1,\cdots,r\}$, we have $\tilde\rho_{p,\cD}(q_i)\ne 1/k$, which implies $\delta_{q_i}(\{\tilde\rho_{p,\cD}=1/k\})=0$. Since the intersection of $\cK$ with $\{k^{\ast},k^{\ast}+1,\cdots\}$ is infinite, we indeed have for infinitely many $k$ that \begin{align*}
         Q\nu(\{\tilde\rho_{p,\cD}=1/k\}) &= \alpha (Q\nu_1)(\{\tilde\rho_{p,\cD}=1/k\}) + (1-\alpha) \sum_{i=1}^r \delta_{q_i}(\{\tilde\rho_{p,\cD}=1/k\}) = 0.
    \end{align*} We repeat the argument to obtain \cref{eq:weak:continuity:intermediate} for this mixture.

    Now, pick any bounded continuous $f$ on $\Lambda$. Then  \begin{align*}
        \int_{\Lambda} f(\lambda) d\hat P_{\Lambda}^i(\lambda) &\to \int_{\cD} \frac{g_f(q)}{\tilde \rho_{p,\cD}(q)}d (Q \nu)(q) \\
        &= \alpha \int_{\cD}\frac{g_f(q)}{\tilde\rho_{p,\cD}(q)} d(Q\nu_1)(q) + (1-\alpha)\sum_{i=1}^r \frac{g_f(q_i)}{\tilde\rho_{p,\cD}(q_i)} \\
        &= \alpha \int_{\Lambda} f(\lambda) d\hat\nu_1(\lambda)+(1-\alpha)\sum_{i=1}^r \omega_i \int_{\Lambda} f(\lambda) dP_p(\lambda| q_i) \\
        &= \int_{\Lambda} f(\lambda) d \hat\nu(\lambda),
    \end{align*} where $\hat\nu =  \alpha \hat\nu_1 + (1-\alpha) \sum_{=1}^r\omega_i P_{p}(\,\cdot\,  | q_i).$ The first line used \cref{eq:weak:continuity:intermediate} and the second split up the  mixture. In the third line, the first integral used the computations in the proof of \cref{theorem:weak:convergence:nice:case}, while the second integral used the definition of $E$, which let us rewrite each $\frac{g_f(q_i)}{\tilde\rho_{p,\cD}(q_i)}$ as in the proof of \cref{theorem:weak:convergence:point:mass}. The final step then consolidated the mixture measure. Hence $\hat P_{\Lambda}^i\stackrel{w}{\to}\hat\nu$.
    \end{proof}

\section{Surface Integral Representations}

Establishing weak convergence in \cref{theorem:weak:convergence:nice:case,theorem:weak:convergence:point:mass} requires writing the disintegration in terms of surface integral. The first lemma below write an integral on $\Lambda$ as a double integral involving a parametrized surface integral on the inside. The second result equates the inner conditional probability integral from the disintegration formula with this surface integral almost everywhere. The first lemma is a simple application of the well-known co-area formula (see, for instance, \cite{nicolaescu2011coarea}). The second applies this co-area type result and argues the inner integrals in both the abstract and surface area versions of the disintegration are two measurable functions that agree on every Borel set, and thus must coincide almost everywhere. Both lemmas impose Assumptions \ref{assumption:compact}, \ref{assumption:Q:smooth}, \ref{assumption:Q:rank}, and  \ref{assumption:prior}.

\begin{lemma} \label{lemma:surface:integral} Let $h$ be an integrable function on $\Lambda$. Then \[\int_{\Lambda} h(\lambda) d\mu_{\Lambda}(\lambda) = \int_{\cD}\int_{Q^{-1}(q)} \frac{h}{\sqrt{\det J_Q J_Q^T}} ds \;d\mu_{\cD}(q),\] where the inner integral is a parametrized surface integral over $Q^{-1}(q)$. If $P_p$ is a prior on $\Lambda$ with density $\rho_p$ in $\mu_{\Lambda}$, then \[\int_{\Lambda} h(\lambda)  dP_p(\lambda) = \int_{\cD}\frac{1}{\tilde\rho_{p,\cD}(q)}\int_{Q^{-1}(q)} \frac{h\cdot\rho_p}{\sqrt{\det J_Q J_Q^T}} ds \;d(QP_p)(q),\]
\end{lemma}
    \begin{proof}
    The first claim follows from the co-area formula \cite{nicolaescu2011coarea}. For the second, observe that \begin{align*}
        \int_{\Lambda} h(\lambda) d P_p(\lambda) &= \int_{\Lambda} h(\lambda) \rho_p(\lambda) d \mu_{\Lambda}(\lambda) = \int_{\cD}\int_{Q^{-1}(q)} \frac{h\cdot\rho_p}{\sqrt{\det J_Q J_Q^T}} ds \;d\mu_{\cD}(q) \\
        &= \int_{\cD}\frac{1}{\tilde\rho_{p,\cD}(q)}\int_{Q^{-1}(q)} \frac{h\cdot\rho_p}{\sqrt{\det J_Q J_Q^T}} ds \;d(QP_p)(q),
    \end{align*} where we used $\frac{dP_p}{d\mu_{\Lambda}}=\rho_p$ and then $\frac{d(QP_p)}{d\mu_{\cD}}=\tilde\rho_{p,\cD}$ (established in the proof of \cite[Theorem 3.5]{esip_2025})  on the right-hand side. 
    \end{proof}

\begin{lemma} \label{lemma:inner:integral:equality} Let $h$ be an integrable function on $\Lambda$. Then for $q$ chosen $Q\mu_{\Lambda}$-a.e., \begin{align*}
    \int_{Q^{-1}(q)} h(\lambda) d \mu_{N}(\lambda| q) =  \frac{1}{\tilde\rho_{\cD}(q)}\int_{Q^{-1}(q)} \frac{h}{\sqrt{\det J_QJ_Q^T}}ds,
\end{align*} where $\tilde\rho_{\cD}(q) = \int_{Q^{-1}(q)} \frac{1}{\sqrt{\det J_QJ_Q^T}}ds$. Next let $P_p$ be any prior with density $\rho_p$ with respect to $\mu_{\Lambda}$. For $q$ chosen $QP_p$-a.e., \begin{align*}
    \int_{Q^{-1}(q)} h(\lambda) d P_p(\lambda| q) = \frac{1}{\tilde\rho_{p,\cD}(q)}\cdot \int_{Q^{-1}(q)} \frac{h \cdot \rho_p}{\sqrt{\det J_QJ_Q^T}}ds.
\end{align*}
\end{lemma}
    \begin{proof} 
    Set $g_1(q)= \int_{Q^{-1}(q)} h(\lambda) d \mu_{N}(\lambda| q)$ and $g_2(q)= \frac{1}{\tilde\rho_{\cD}(q)}\int_{Q^{-1}(q^{\ast})} \frac{h}{\sqrt{\det J_QJ_Q^T}}ds$. Pick any $B\in\cB_{\cD}$. Then  applying the disintegration theorem, \begin{align*}
        \int_B g_1(q) d (Q\mu_{\Lambda})(q) &= \int_{\cD} \chi_B(q)\left[\int_{Q^{-1}(q)} h(\lambda) d\mu_N(\lambda| q) \right]d(Q\mu_{\Lambda})(q) \\
        &= \int_{\cD} \left[\int_{Q^{-1}(q)}  \chi_B(Q(\lambda))\cdot h(\lambda) d\mu_N(\lambda| q) \right]d(Q\mu_{\Lambda})(q) \\
        &= \int_{\Lambda} \chi_B(Q(\lambda)) h(\lambda) d\mu_{\Lambda}(\lambda).
    \end{align*} Next, applying that $\tilde\rho_{\cD} = \frac{d (Q\mu_{\Lambda})}{d\mu_{\cD}}$ along with the first claim in \cref{lemma:surface:integral}, \begin{align*}
        \int_B g_2(q) d (Q\mu_{\Lambda})(q) &= \int_{\cD} \frac{\chi_B(q)}{\tilde\rho_{\cD}(q)}\left[\int_{Q^{-1}(q)} \frac{h}{\sqrt{\det J_QJ_Q^T}}ds\right]d(Q\mu_{\Lambda})(q) \\
        &= \int_{\cD} \frac{1}{\tilde\rho_{\cD}(q)}\left[\int_{Q^{-1}(q)}  \frac{(\chi_B\circ Q)\cdot  h}{\sqrt{\det J_QJ_Q^T}}ds \right]d(Q\mu_{\Lambda})(q) \\
        &= \int_{\cD} \left[\int_{Q^{-1}(q)}  \frac{(\chi_B\circ Q)\cdot  h}{\sqrt{\det J_QJ_Q^T}}ds \right]d\mu_{\cD}(q) \\
        &= \int_{\Lambda} \chi_B(Q(\lambda)) h(\lambda)d\mu_{\Lambda}(\lambda).
    \end{align*} This proves $g_1(q)=g_2(q)$ for $q$ chosen $Q\mu_{\Lambda}$-a.e.
     
    The other claim follows a similar logic. Set $g_1(q)=\int_{Q^{-1}(q)} h(\lambda) d P_p(\lambda| q) $ and $g_2(q)=\frac{1}{\tilde\rho_{p,\cD}(q)}\cdot \int_{Q^{-1}(q^{\ast})} \frac{h \rho_p}{\sqrt{\det J_QJ_Q^T}}ds$. Pick any $B\in\cB_{\cD}$. Then observe that by disintegration \begin{align*}
            \int_B g_1(q) d (QP_p)(q) &= \int_{\cD} \chi_B(q)  \left[\int_{Q^{-1}(q)} h(\lambda) d P_p(\lambda| q) \right]\;d(QP_p)(q) \\
            &= \int_{\cD}  \left[ \int_{Q^{-1}(q)}\chi_B(Q(\lambda)) h(\lambda) d P_p(\lambda| q)\right]  \; d(QP_p)(q)\\
            &= \int_{\Lambda} \chi_B(Q(\lambda)) h(\lambda) d P_p(\lambda).
        \end{align*} 

        By similar logic but using the second result in \cref{lemma:surface:integral}, we have \begin{align*}
            \int_B g_2(q) d (QP_p)(q) &= \int_{\cD} \chi_B(q)  \left[\frac{1}{\tilde\rho_{p,\cD}(q)}\cdot \int_{Q^{-1}(q^{\ast})} \frac{h \cdot \rho_p}{\sqrt{\det J_QJ_Q^T}}ds \right]\;d(QP_p)(q) \\
            &= \int_{\cD}  \left[\frac{1}{\tilde\rho_{p,\cD}(q)}\cdot \int_{Q^{-1}(q^{\ast})} \frac{(\chi_B\circ Q)\cdot h \cdot\rho_p}{\sqrt{\det J_QJ_Q^T}}ds \right]\;d(QP_p)(q)\\
            &= \int_{\Lambda} \chi_B(Q(\lambda)) h(\lambda) d P_p(\lambda).
        \end{align*} 

        Thus, for all $B\in\cB_{\cD}$, $\int_{B}g_1(q) d (QP_p)(q)=\int_B g_2(q) d(QP_p)(q)$. Therefore, $g_1(q)=g_2(q)$ for $q$ chosen  $QP_p$-a.e.
    \end{proof}

\bibliographystyle{siamplain}
\bibliography{ref}

\end{document}